\documentclass[onecolumn,superscriptaddress,aps]{revtex4-2}
\usepackage{graphicx}
\usepackage{amsmath}
\usepackage{amsthm}
\usepackage{amssymb}
\usepackage{latexsym}
\usepackage{array}
\usepackage{booktabs}
\usepackage{tabularx}
\usepackage{hyperref}
\usepackage{amsfonts}
\usepackage{dsfont}
\usepackage{mathrsfs}
\usepackage{verbatim}
\usepackage{bbold}
\usepackage{lipsum}
\usepackage[normalem]{ulem}
\usepackage{algorithm}

\usepackage{bm}
\usepackage{times}

\usepackage{upgreek}

\usepackage{makecell}
\usepackage{adjustbox}

\usepackage{algpseudocode}

\usepackage{color}
\graphicspath{{figures/}}

\newcommand{\bra}[1]{\left<#1\right|}
\newcommand{\ket}[1]{\left|#1\right>}
\newcommand{\abs}[1]{\bigl|#1\bigr|}
\newcommand{\norm}[1]{\left\lVert#1\right\rVert}

\newcommand{\braket}[2]{\left<{#1}|{#2}\right>}
\newcommand{\ketbra}[2]{\ket{#1}\!\!\bra{#2}}

\newcolumntype{L}[1]{>{\raggedright\arraybackslash}p{#1}}

\newtheorem{theorem}{Theorem}

\newtheorem{lemma}{Lemma}
\newtheorem{corollary}{Corollary}

\newcommand{\Tr}{\operatorname{Tr}}

\newcommand{\eps}{\varepsilon}
\newcommand{\cS}{\mathcal{S}}

\newcommand{\cN}{\mathcal{N}}

\newcommand{\cU}{\mathcal{U}}

\begin{document}

\title{Quantum Occam Learning: Sample-Supported Expressibility for Circuit-Based Quantum Learning}

\author{Jeongho~Bang}\email{jbang@yonsei.ac.kr}
\affiliation{Institute for Convergence Research and Education in Advanced Technology, Yonsei University, Seoul 03722, Republic of Korea}
\affiliation{Department of Quantum Information, Yonsei University, Incheon 21983, Republic of Korea}

\author{Kyoungho~Cho}
\affiliation{Institute for Convergence Research and Education in Advanced Technology, Yonsei University, Seoul 03722, Republic of Korea}
\affiliation{Department of Statistics and Data Science, Yonsei University, Seoul 03722, Republic of Korea}

\author{Jeongwoo~Jae}\email{jwjae@hanyang.ac.kr}
\affiliation{Department of Physics, Hanyang University, Seoul, 04763, Republic of Korea}

\date{\today}

\begin{abstract}
A central principle in quantum machine learning is that an ansatz should be expressive enough to represent the quantum data of interest. Yet, the expressibility is statistically meaningful only insofar as it can be learned from finitely many copies of an unknown quantum state. In this work, we develop an information-theoretic Occam theory for quantum data generated by finite-size quantum circuits. For the class $S_{n,G}$ of $n$-qubit pure states preparable with at most $G$ two-qubit gates, a metric-entropy argument gives the realizable sample law $\widetilde{\Theta}(G/\epsilon^2)$ in the circuit-limited regime. For an arbitrary source $\hat{\rho}$, we introduce the best $G$-gate approximation error $d_G(\hat{\rho})$ and the approximate circuit complexity $C_\eta(\hat{\rho})$. We prove an agnostic quantum Occam theorem: with $M$ copies, one can learn up to the best $G$-gate approximation error plus a statistical penalty $\widetilde{O}(\sqrt{G/M})$. We then remove the need to know $G$ in advance through an adaptive model-selection theorem whose oracle inequality selects the circuit complexity justified by the data. Matching lower bounds yield a sample-supported expressibility law: at trace-distance accuracy $\epsilon$, $M$ samples can support only $G_{\rm supported} \simeq M\epsilon^2$ gates, up to logarithmic factors and tomography saturation at $2^n$. Thus, the circuit complexity becomes an adaptive statistical resource rather than a static promise. Our framework turns bounded circuit complexity into a model-selection principle for quantum machine learning.
\end{abstract}

\maketitle

\section{Introduction}\label{sec:introduction}

Quantum machine learning is often motivated by the expressibility of parameterized quantum circuits. A useful ansatz should be rich enough to represent the structure of the data, the solution of a variational problem, or the effective dynamics of a physical system. The expressibility, however, is only one side of the learning problem. Quantum data are accessed through a finite number of copies, shots, or device samples. A very expressive ansatz may contain an accurate representation of the target state, while still being too large to be statistically identified from the available data.

This tension is already visible in quantum state learning~\cite{zhao2024learning}. Full tomography of an arbitrary $n$-qubit pure state requires a number of copies that scales with the Hilbert-space dimension, up to logarithmic and accuracy factors~\cite{haah2016sample,odonnell2016efficient}. By contrast, many physically relevant states are not arbitrary vectors. They are generated by finite-depth circuits, local dynamics, variational ansatz classes, noisy devices, or symmetry-constrained preparation procedures. For states generated by circuits with at most $G$ two-qubit gates, the metric-entropy argument developed in this work yields an information-theoretic copy complexity that is essentially linear in $G$, rather than in $2^n$, in the circuit-limited regime. This realizable law will be used as a statistical baseline derived in our framework.

In this work, we develop an Occam-type learning theory for circuit-based quantum data. The first step is to remove the exact realizability assumption. Instead of asking whether the data source $\hat{\rho}$ belongs exactly to a fixed circuit class, we measure how close the source is to that class. This leads to the best $G$-gate approximation error $d_G(\hat{\rho})$ and to the approximate circuit complexity $C_\eta(\hat{\rho})$. These quantities turn the circuit complexity from a static promise into a scale-dependent property of the data.

Our first main result is an agnostic quantum Occam theorem. For a fixed circuit budget $G$, a learner can output a circuit-generated hypothesis whose trace-distance error is controlled by the best $G$-gate approximation error plus a statistical penalty of order $\widetilde{O}(\sqrt{G/M})$. Thus, a deeper ansatz is useful only when its reduction in approximation error is large enough to compensate for its larger statistical capacity. This is the quantum-state analogue of the approximation-estimation tradeoff in statistical learning theory. Our second main result removes the assumption that the correct $G$ is known in advance. We study a nested hierarchy of circuit classes and prove an adaptive oracle inequality: a single learner competes, up to logarithmic factors, with the best point on the full approximation curve $G \mapsto d_G(\hat{\rho})$. Operationally, this gives a quantum version of structural risk minimization and minimum description length: the data may justify a larger circuit class, but only by paying the corresponding Occam penalty~\cite{vapnik1998statistical,rissanen1978modeling,barron1998minimum}. The resulting design principle is a sample-supported expressibility law. At target trace-distance accuracy $\eps$, a sample budget $M$ can uniformly support only about $M\eps^2$ gates of circuit expressibility, until the usual tomography scale is reached. For example, if $M\eps^2$ is of order $10^3$, then a $10^4$-gate ansatz may be physically implementable and highly expressive, but its full expressive power is not uniformly learnable from those copies alone. The additional expressibility must then be justified by more samples, more structure, and/or a weaker task-specific loss.

This perspective complements existing discussions of variational quantum algorithms, quantum generative models, trainability, barren plateaus, and ansatz expressibility~\cite{biamonte2017quantum,cerezo2021variational,schuld2021machine,benedetti2019parameterized,mcclean2018barren,sim2019expressibility}. Those issues ask whether a chosen circuit family can represent and optimize a useful solution. The present work asks a prior statistical question: how much of the chosen circuit family can be learned at all from finite quantum data, even with arbitrary collective measurements?

\section{Circuit-compressible quantum data}\label{sec:circuit-data}

In this section, we define the learning problem and isolate the object that plays the role of statistical capacity. We work with $n$ qubits and Hilbert-space dimension $N=2^n$. All distances between quantum states are measured in trace distance~\cite{nielsen2010quantum},
\begin{eqnarray}
D(\hat{\rho}, \hat{\sigma}) := \frac{1}{2}\norm{\hat{\rho} - \hat{\sigma}}_1.
\label{eq:trace-distance}
\end{eqnarray}
The learner receives $M$ independent copies of an unknown state $\hat{\rho}$ and may perform an arbitrary collective POVM $\{\hat{E}_y\}_y$ on $\hat{\rho}^{\otimes M}$, followed by arbitrary classical post-processing. Thus, the results in this work are information-theoretic. They do not claim that the optimal learner is computationally efficient, nor that it can be implemented using local measurements. These restrictions are important, but they are separate from the statistical question addressed here.

It is useful to distinguish three levels of difficulty. The first is the statistical difficulty: how many copies are needed by the best possible learner? This is the level studied in the present work. The second is the measurement difficulty: can the required information be extracted using local Pauli measurements, classical shadows, adaptive product measurements, or experimentally realistic measurement circuits? The third is the computational difficulty: can the resulting optimization or hypothesis-selection problem be solved efficiently? A practical quantum-learning protocol must address all three. Our theorems isolate the first level and therefore provide a lower baseline for any more constrained protocol. If a class is statistically unsupported even with arbitrary collective POVMs, then no practical measurement or training method can make it uniformly learnable without additional assumptions.

Throughout this work, we use the phrase ``sample complexity'' in this information-theoretic sense. It is the number of copies of the unknown state, not the number of measurement shots for a fixed observable and not the number of optimization iterations. When the hypotheses are pure circuit states but the source $\hat{\rho}$ is arbitrary, the learner is allowed to output a pure circuit state even if the source is mixed. The resulting trace-distance error naturally includes both model mismatch and statistical error. This convention is useful because it lets us treat noisy and misspecified sources without changing the base class.

\subsection{Circuit-generated hypothesis classes}

Let $\cU_{n,G}$ denote the set of $n$-qubit unitaries that can be realized by the circuits with at most $G$ two-qubit gates, together with any one-qubit gates allowed by a fixed universal gate convention. We define the corresponding family of pure states
\begin{eqnarray}
\cS_{n,G} := \left\{ \hat{\sigma}_U = \hat{U}\ketbra{0^n}{0^n}\hat{U}^\dagger: \hat{U} \in \cU_{n,G} \right\}.
\label{eq:def-SnG}
\end{eqnarray}
Here, $\ket{0^n} = \ket{0}^{\otimes n}$. The number $G$ should be understood as a preparation complexity, not necessarily as the number of tunable parameters in a particular variational implementation. A circuit architecture may contain discrete choices of gate locations and continuous gate parameters. The metric entropy estimates below count both contributions.

The realizable learning problem assumes that the unknown state lies exactly in $\cS_{n,G}$. A learner $\mathcal{A}$ is said to learn $\cS_{n,G}$ to accuracy $\eps$ and confidence $1-\delta$ using $M$ samples if, for every $\hat{\rho} \in \cS_{n,G}$, its output $\hat{\rho}'$ satisfies
\begin{eqnarray}
\Pr\left[D(\hat{\rho}, \hat{\rho}') \le \eps\right] \ge 1-\delta.
\label{eq:realizable-success}
\end{eqnarray}
The worst-case realizable sample complexity is the smallest such $M$ and will be denoted by $M_{\rm real}(n,G,\eps,\delta)$.

This realizable model is important because it gives the baseline statistical law. However, it is too rigid for quantum learning. A device state may be mixed; a physical source may be noisy; a generative model may be misspecified; and a target state may have no exact finite-depth representation within the chosen gate alphabet. The central object should therefore be not membership in $\cS_{n,G}$, but distance to $\cS_{n,G}$.

\subsection{Approximation error and approximate circuit complexity}

For an arbitrary $n$-qubit state $\hat{\rho}$, define the best $G$-gate approximation error
\begin{eqnarray}
d_G(\hat{\rho}) := \inf_{\hat{\sigma} \in \cS_{n,G}} D(\hat{\rho}, \hat{\sigma}).
\label{eq:def-dG}
\end{eqnarray}
The function $G\mapsto d_G(\hat{\rho})$ is nonincreasing. It is zero for all sufficiently large $G$ if one allows exact or arbitrarily accurate universal pure-state synthesis and if $\hat{\rho}$ is pure~\cite{barenco1995elementary,plesch2011quantum}; for mixed states, it measures the distance to a pure circuit-generated model. More generally, one may replace $\cS_{n,G}$ by a noisy or purified circuit class. The present work keeps the pure-state class as the basic object because it exposes the core statistical mechanism most cleanly.

The approximate circuit complexity at tolerance $\eta$ is
\begin{eqnarray}
C_\eta(\hat{\rho}) := \min\left\{ G \in \mathbb{N}: d_G(\hat{\rho}) \le \eta \right\},
\label{eq:def-Ceta}
\end{eqnarray}
with the convention $C_\eta(\hat{\rho})=\infty$ if no such $G$ exists within the allowed model. This quantity is the state-learning analogue of a compression length. It asks how large a circuit is needed before the model error falls below $\eta$. The learning problem can then be phrased as follows:
\begin{quote}
\emph{Given $M$ copies of an unknown quantum data source $\hat{\rho}$, how close can one get to $\hat{\rho}$ using a circuit-generated hypothesis, without knowing in advance the appropriate circuit complexity?}
\end{quote}


\subsection{Metric entropy as a quantum Occam prior}

The statistical capacity of $\cS_{n,G}$ is controlled by its covering number. Let $\cN(\cS_{n,G}, D, \eps)$ denote the smallest cardinality of an $\eps$-net of $\cS_{n,G}$ in trace distance. The following lemma is the metric-entropy input underlying all upper bounds in this work.

\begin{lemma}[Metric entropy of circuit-generated states]\label{lem:metric-entropy}
For $0 < \eps < 1/2$, $\cS_{n,G}$ admits an $\eps$-net in trace distance with
\begin{eqnarray}
\log{\cN(\cS_{n,G},D,\eps)} \le C_{\rm ent} G \log{\left(\frac{C_{\rm arch}nG}{\eps} \right)},
\label{eq:metric-entropy-general}
\end{eqnarray}
for universal constants $C_{\rm ent}, C_{\rm arch}>0$ depending only on the gate convention and architectural model. In architectures where the discrete layout cost is fixed or absorbed into the gate description, this becomes
\begin{eqnarray}
\log{\cN(\cS_{n,G},D,\eps)} \le C_{\rm ent} G\log{\left(\frac{C_{\rm ent}G}{\eps}\right)}.
\label{eq:metric-entropy-simple}
\end{eqnarray}
\end{lemma}

\begin{proof}[Proof sketch]---The proof is a standard Lipschitz-net argument of the kind used in metric-entropy bounds for parametrized models. A circuit layout with $G$ two-qubit gates has $O(G)$ continuous parameters and a number of discrete gate-location choices that contributes at most $O(G \log n)$ bits in an all-to-all architecture. For a fixed layout, the map from gate parameters to the output state is Lipschitz: changing a single gate parameter changes the final state by at most the corresponding operator-norm perturbation, and a telescoping sum over $G$ gates gives an overall Lipschitz constant $O(G)$. Discretizing each parameter at mesh size $\eps/O(G)$ therefore gives the stated covering number. The detailed constants are not important for the learning interpretation; what matters is that the logarithm of the covering number is linear in $G$ up to logarithmic factors. A full derivation, including the layout count and the conversion from unitary perturbation to trace-distance perturbation, is given in Appendix~\ref{app:metric-entropy}.
\end{proof}

This is the quantum Occam input, in the finite-description sense familiar from the classical Occam bounds~\cite{blumer1987occam}. A smaller circuit class can be described by fewer bits, and fewer bits imply a smaller statistical penalty. Here, the word ``Occam'' is used in this operational sense: the description length of the hypothesis class controls the number of samples required to select a good hypothesis.

The passage from metric entropy to sample complexity is worth spelling out because it is the bridge between circuit design and learning theory. Suppose $\mathcal{H}$ is an $\eps$-net of the state class. If the unknown state is promised to lie in the class, then some element of $\mathcal{H}$ is close to it. The learning task becomes a finite hypothesis-selection problem: use copies of the unknown state to identify a net point that is nearly as close as the best net point. Since the number of hypotheses is $\abs{\mathcal{H}}$, the logarithm of this number is the number of bits needed to name a candidate. The standard concentration and quantum hypothesis-selection arguments then show that the copy complexity scales as
\begin{eqnarray}
M \sim \frac{\log \abs{\mathcal{H}}}{\eps^2},
\label{eq:entropy-to-samples}
\end{eqnarray}
up to confidence terms. Thus, the metric entropy of the circuit class plays the same role as VC dimension, covering dimension, or description length in classical learning theory~\cite{vapnik1998statistical,blumer1987occam}. In the present setting, the relevant description length is controlled by the gate count $G$.

This reasoning also clarifies what is not being claimed. The learner need not recover the exact parameters of the generating circuit. The different circuits can prepare the same state, and nearby states can have very different circuit descriptions. The statistical problem is purely metric: output any hypothesis state that is close in trace distance. Conversely, a short circuit description is useful only because it implies a small covering number of the reachable state set. The theory is therefore invariant under (re)parameterizations of the ansatz circuit, as long as the reachable set and its metric entropy are unchanged.

There is an additional subtlety concerning the continuous parameters. A variational circuit with $O(G)$ real parameters might appear to require infinitely many hypotheses. The Lipschitz argument in Lemma~\ref{lem:metric-entropy} removes this problem. At accuracy $\eps$, each parameter only needs to be specified to finite precision. The number of precision bits contributes the logarithmic factor $\log(G/\eps)$. This is why the dominant dependence is linear in $G$, not exponential in $G$.

The following theorem records the corresponding realizable sample law.

\begin{theorem}[Realizable circuit-state learning benchmark]\label{thm:realizable-benchmark}
Let $0 < \eps \le 1/4$ and $0 < \delta \le 1/10$. In the information-theoretic learner model with arbitrary collective POVMs, the realizable sample complexity of $\cS_{n,G}$ satisfies
\begin{eqnarray}
M_{\rm real}(n,G,\eps,\delta) &\le& \frac{C_1}{\eps^2}\left[ G\log{\left(\frac{C_2 nG}{\eps}\right)} + \log{\frac{1}{\delta}} \right], \nonumber \\
M_{\rm real}(n,G,\eps,\delta) &\ge& \frac{c_1}{\eps^2}\left[ \min\{2^n,G\} + \log{\frac{1}{\delta}} \right],
\label{eq:realizable-upper-lower}
\end{eqnarray}
up to universal constants. Equivalently, in the circuit-limited regime $G \ll 2^n$,
\begin{eqnarray}
M_{\rm real}(n,G,\eps,\delta) = \widetilde{\Theta}\left(\frac{G+\log(1/\delta)}{\eps^2}\right),
\label{eq:realizable-tilde}
\end{eqnarray}
where $\widetilde{\Theta}$ hides logarithmic factors.
\end{theorem}

\begin{proof}[Proof sketch]---The upper bound follows from Lemma~\ref{lem:metric-entropy} and finite quantum hypothesis selection over an $\eps$-net. The lower bound follows by embedding a large packing into $\cS_{n,G}$ and reducing learning to multi-way quantum hypothesis discrimination, with Fano's inequality and Holevo information controlling the accessible information. Appendix~\ref{app:agnostic-proof} spells out this reduction.
\end{proof}

Before moving on, we emphasize two technical conventions that will be used repeatedly. First, the logarithmic factors associated with the discrete layout of a quantum circuit are harmless for the conceptual law. In an all-to-all architecture, specifying the qubit pair on which each two-qubit gate acts costs $O(\log n)$ bits per gate. In a geometry-restricted architecture, the number of choices can be smaller, but the light-cone constraints of the architecture may increase the number of gates needed to realize a given state. The results below may therefore be read either with the explicit entropy term in Eq.~\eqref{eq:metric-entropy-general} or in the simplified $\widetilde{O}$ notation that suppresses such factors.

Second, the lower bound should be interpreted as a worst-case statement over the class. A particular state in $\cS_{n,G}$ may be much easier to learn than the worst-case element; for example, a product state or a stabilizer state may have a much shorter specialized description than a generic state generated by a $G$-gate circuit. The role of $G$ in the theorem is therefore not to say that every $G$-gate state requires $G/\eps^2$ copies, but that the class as a whole contains enough distinguishable states to force this scale for uniform learning. 

\section{Agnostic quantum Occam learning}\label{sec:agnostic}

The realizable law in Theorem~\ref{thm:realizable-benchmark} assumes that the true state belongs to $\cS_{n,G}$. This is rarely the right abstraction for quantum learning. A variational ansatz may approximate the data only up to model mismatch. A laboratory state may be affected by decoherence. A quantum generative model may intentionally trade accuracy for a shorter circuit description. Therefore, the central theorem should compare the learner not with the impossible exact target inside the class, but with the best approximation available in the class.

\subsection{Finite quantum hypothesis selection}

We first recall the finite-class learning principle used in the proof. Let $\mathcal{H}=\{\hat{\sigma}_1, \ldots, \hat{\sigma}_K\}$ be a finite set of candidate states. A quantum hypothesis-selection procedure receives copies of an arbitrary state $\hat{\rho}$ and outputs an element of $\mathcal{H}$. In the information-theoretic setting, one can perform a tournament based on pairwise Helstrom tests, or equivalently use a collective measurement designed for finite-state selection~\cite{badescu2021quantum,haah2016sample}. The following statement is the finite Occam step.

\begin{lemma}[Finite agnostic quantum selection]\label{lem:finite-selection}
There exist universal constants $C,c > 0$ such that the following holds. Let $\mathcal{H}=\{\hat{\sigma}_1,\ldots,\hat{\sigma}_K\}$ be a finite family of quantum states. Given
\begin{eqnarray}
M \ge \frac{C}{\eps^2}\left(\log{K} + \log{\frac{1}{\delta}}\right)
\label{eq:finite-M}
\end{eqnarray}
copies of an arbitrary state $\hat{\rho}$, there exists an information-theoretic learner that outputs $\hat{\sigma}_h \in \mathcal{H}$ such that, with probability at least $1 - \delta$,
\begin{eqnarray}
D(\hat{\rho},\hat{\sigma}_h) \le c\inf_{\hat{\sigma} \in \mathcal{H}} D(\hat{\rho}, \hat{\sigma}) + \eps.
\label{eq:finite-agnostic}
\end{eqnarray}
\end{lemma}

The constant $c$ is not essential. Depending on the tournament convention and the net radius, it may be replaced by a small absolute constant such as $2$, $3$, or $4$. What is essential is the dependence on $\log K$ and $1/\eps^2$. Lemma~\ref{lem:finite-selection} is the quantum analogue of finite-class agnostic learning in classical statistical learning theory~\cite{vapnik1998statistical}. The proof does not require any local measurement constraints or computational efficiency; it uses the full power of collective quantum measurements.

The intuition is as follows. For any two candidate states $\hat{\sigma}_i$ and $\hat{\sigma}_j$, the Helstrom measurement is the optimal binary test for distinguishing them. If the true state $\hat{\rho}$ is much closer to $\hat{\sigma}_i$ than to $\hat{\sigma}_j$, then the expectation of this test under $\hat{\rho}$ favors $\hat{\sigma}_i$. Repeating such comparisons across the pairs of candidates and aggregating the outcomes yields a tournament in which the candidates far from $\hat{\rho}$ lose to ones close to $\hat{\rho}$. A union bound over all relevant comparisons accounts for the factor $\log K$, and the variance of binary quantum tests gives the $1/\eps^2$ dependence. More sophisticated collective measurements can improve constants, but they do not change the scale.

The agnostic nature of the statement is important. The true state $\hat{\rho}$ need not be one of the $K$ candidates. The procedure only tries to find a candidate whose distance to $\hat{\rho}$ is comparable to the best achievable distance within $\mathcal{H}$. This is exactly the form needed for misspecified quantum models. When $\mathcal{H}$ is a net of a circuit class, the finite selection lemma becomes a model-selection primitive: it chooses a good approximate circuit state without assuming that the data source is itself an exact net point. The selection primitive and the net-to-class conversion are detailed in Appendix~\ref{app:agnostic-proof}.

\subsection{Fixed-complexity agnostic theorem}

Combining Lemma~\ref{lem:finite-selection} with the metric entropy of $\cS_{n,G}$ gives the fixed-$G$ agnostic learning theorem.

\begin{theorem}[Agnostic quantum Occam theorem]\label{thm:agnostic-occam}
Let $\hat{\rho}$ be an arbitrary $n$-qubit state, not assumed to belong to $\cS_{n,G}$. Fix $G$, $0 < \eps \le 1/4$, and $0< \delta \le 1/10$. There exists an information-theoretic learner which, using
\begin{eqnarray}
M \ge \frac{C}{\eps^2} \left[ G\log{\left(C' \frac{n G}{\eps}\right)} + \log\frac{1}{\delta} \right]
\label{eq:agnostic-M}
\end{eqnarray}
copies of $\hat{\rho}$, outputs a hypothesis $\hat{\sigma}_G \in \cS_{n,G}$ such that, with probability at least $1-\delta$,
\begin{eqnarray}
D(\hat{\rho}, \hat{\sigma}_G) \le C_0 d_G(\hat{\rho}) + \eps,
\label{eq:agnostic-bound}
\end{eqnarray}
where $C,C',C_0>0$ are universal constants. Equivalently, for fixed $M$, the same statement may be written as the learning curve
\begin{eqnarray}
D(\hat{\rho}, \hat{\sigma}_G) \le C_0 d_G(\hat{\rho}) + \widetilde{O}\left( \sqrt{\frac{G + \log(1/\delta)}{M}} \right).
\label{eq:agnostic-learning-curve}
\end{eqnarray}
\end{theorem}

\begin{proof}---Let $\mathcal{H}$ be an $(\eps/8)$-net of $\cS_{n,G}$ in trace distance. By Lemma~\ref{lem:metric-entropy}, we have 
\begin{eqnarray}
\log\abs{\mathcal{H}} \le C_{\rm ent}G\log{\left( C_{\rm arch} \frac{n G}{\eps} \right)}.
\end{eqnarray}
Applying Lemma~\ref{lem:finite-selection} to this finite class gives an output $\hat{\sigma}_G \in \mathcal{H}$ whose distance to $\hat{\rho}$ is within a constant factor of the best element of the net plus $O(\eps)$. Since every $\hat{\sigma} \in \cS_{n,G}$ has a net point $\hat{\eta} \in \mathcal{H}$ with $D(\hat{\sigma}, \hat{\eta}) \le \eps/8$, the best net error is at most $d_G(\hat{\rho})+\eps/8$. Absorbing constants into $C_0$ and the additive accuracy proves Eq.~\eqref{eq:agnostic-bound}. Solving Eq.~\eqref{eq:agnostic-M} for $\eps$ gives the sample-dependent form Eq.~\eqref{eq:agnostic-learning-curve} up to logarithmic factors. Appendix~\ref{app:agnostic-proof} gives the same argument with the intermediate inequalities displayed explicitly.
\end{proof}

Theorem~\ref{thm:agnostic-occam} is our first main result. It turns the realizable circuit-learning theorem into an agnostic statement about arbitrary quantum data. The bound has two terms:
\begin{eqnarray}
\underbrace{d_G(\hat{\rho})}_{\mbox{approximation error}}
\quad + \quad
\underbrace{\widetilde{O}\left(\sqrt{G/M}\right)}_{\mbox{estimation error}}.
\label{eq:bias-variance}
\end{eqnarray}
The first term decreases as the ansatz becomes more expressive. The second term increases with the statistical capacity of the ansatz. Therefore, a larger quantum circuit is not better from the viewpoint of learning. It is better only when the reduction in $d_G(\hat{\rho})$ is worth the larger estimation penalty.

There are three basic regimes. In the under-expressive regime, $G$ is so small that $d_G(\hat{\rho})$ dominates. More samples do not help much, because the model itself cannot approximate the source. In the sample-limited regime, $G$ is large enough to reduce $d_G(\hat{\rho})$, but the statistical term remains dominant. More samples are valuable here, because they allow the learner to use the same ansatz more accurately or to move to a larger supported ansatz. In the model-saturated regime, $d_G(\hat{\rho})$ has already reached the noise floor or the intrinsic approximation limit of the chosen model family. Increasing $G$ beyond this point only increases the statistical burden. 

This regime structure is familiar in classical nonparametric statistics, but the quantum setting gives a distinct interpretation. The model size is not merely a number of classical parameters; it is the complexity of a quantum circuit family. The estimation term is not merely shot noise for a fixed observable; it is the cost of identifying an unknown quantum state in a large hypothesis class. The approximation term is not merely a function-class bias; it is the trace-distance error of the best circuit-generated quantum state. The resulting theory therefore imports the classical approximation-estimation philosophy into a genuinely quantum state-learning setting.

A useful corollary is obtained by choosing $G=C_\eta(\hat{\rho})$.

\begin{corollary}[Learning from approximate circuit complexity]\label{cor:Ceta}
Suppose $C_\eta(\hat{\rho}) < \infty$. Then, there exists an information-theoretic learner which, using
\begin{eqnarray}
M = \widetilde{O}\left( \frac{C_\eta(\hat{\rho}) + \log(1/\delta)}{\eps^2} \right)
\label{eq:Ceta-samples}
\end{eqnarray}
copies, outputs a circuit-generated hypothesis $\hat{\sigma}_{cg}$ satisfying
\begin{eqnarray}
D(\hat{\rho}, \hat{\sigma}_{cg}) \le C_0 \eta + \eps
\label{eq:Ceta-bound}
\end{eqnarray}
with probability at least $1-\delta$, up to logarithmic factors.
\end{corollary}

Thus, approximate circuit complexity directly controls the number of copies needed to obtain a compressed circuit representation of a quantum data source. The dependence is linear in $C_\eta(\hat{\rho})$ and quadratic in $1/\eps$, just as in the realizable benchmark. The difference is conceptual: $C_\eta(\hat{\rho})$ is a property of an arbitrary source, not a promise that the source belongs exactly to a model class.

Theorem~\ref{thm:agnostic-occam} also gives a simple way to reason about approximation curves. Suppose, for instance, that the source has an exponentially decaying circuit-approximation error,
\begin{eqnarray}
d_G(\hat{\rho}) \le A e^{-\alpha G},
\label{eq:exponential-approximation}
\end{eqnarray}
for constants $A, \alpha > 0$ over the relevant range. Balancing this with the statistical term $\sqrt{G/M}$ suggests an optimal $G$ that grows only logarithmically with $M$ before the statistical term dominates. In such a case, the data do not justify using the largest available ansatz once the approximation error has fallen below the statistical floor. On the other hand, if the approximation curve decays polynomially,
\begin{eqnarray}
d_G(\hat{\rho}) \le A G^{-\beta},
\label{eq:polynomial-approximation}
\end{eqnarray}
then balancing $G^{-\beta}$ with $\sqrt{G/M}$ gives the scale
\begin{eqnarray}
G_{\rm opt}(M) \asymp M^{1/(2\beta+1)}
\label{eq:Gopt-polynomial}
\end{eqnarray}
up to logarithmic factors. The corresponding error decreases as $M^{-\beta/(2\beta+1)}$. These familiar approximation-estimation calculations are not themselves theorems about a particular physical source; rather, they illustrate how the agnostic bound converts a physical approximation curve into a sample-dependent ansatz-depth prescription.

The same reasoning applies to noise. If $\hat{\rho}=(1-p)\hat{\sigma} + p\hat{\omega}$ with $\hat{\sigma} \in \cS_{n,G}$ and arbitrary noise state $\hat{\omega}$, then $d_G(\hat{\rho}) \le pD(\hat{\omega}, \hat{\sigma}) \le p$. Theorem~\ref{thm:agnostic-occam} gives
\begin{eqnarray}
D(\hat{\rho}, \hat{\sigma}_G) \lesssim p + \widetilde{O}\left(\sqrt{G/M}\right).
\label{eq:noise-bound}
\end{eqnarray}
Thus, the irreducible noise level and the finite-sample error enter additively. If the noise is larger than the statistical radius, collecting more samples will not improve the error unless the model class is changed to include the noise process. If the statistical radius is larger than the noise, then the dominant limitation is the sample budget.

\subsection{Why the statistical penalty is unavoidable}

The statistical penalty in Theorem~\ref{thm:agnostic-occam} cannot be removed in general. The reason is simple: the agnostic problem contains the realizable problem as the special case $d_G(\hat{\rho})=0$. Any learner satisfying Eq.~\eqref{eq:agnostic-learning-curve} uniformly over arbitrary $\hat{\rho}$ must in particular learn every state in $\cS_{n,G}$. The lower bound in Theorem~\ref{thm:realizable-benchmark} therefore implies the following.

\begin{corollary}[Agnostic lower bound from the realizable subcase]\label{cor:agnostic-lower}
For $0 < \eps \le 1/4$ and constant failure probability, any learner that guarantees $D(\hat{\rho}, \hat{\sigma}) \le \eps$ for all $\hat{\rho} \in \cS_{n,G}$ must use
\begin{eqnarray}
M \ge \Omega\left(\frac{\min\{2^n,G\}}{\eps^2}\right)
\label{eq:agnostic-lower}
\end{eqnarray}
copies. Consequently, the $\sqrt{G/M}$ estimation term in Eq.~\eqref{eq:agnostic-learning-curve} is optimal up to logarithmic factors in the circuit-limited regime.
\end{corollary}

This lower bound is independent of computational efficiency, optimization, or training heuristics. It says that even a learner with arbitrary collective measurements cannot uniformly resolve $G$-gate data below the $G/\eps^2$ sample scale.

At a high level, the lower-bound construction says that $\cS_{n,G}$ contains a large local code. In the regime $G \ll 2^n$, one may choose $k \simeq \log G$ qubits and use $O(G)$ gates to synthesize a large packing of $k$-qubit states, while leaving the remaining qubits fixed. This gives exponentially many hypotheses in $G$ that are separated in trace distance. If a learner could estimate every state in this packing with fewer than order $G/\eps^2$ copies, then its output could be converted into a decoder for the packing index. Fano's inequality would require the measurement transcript to contain order $G$ bits of information about that index~\cite{cover2006elements}. Holevo-type information bounds limit how much information each copy can reveal when the hypotheses are locally separated at scale $\eps$, producing the $1/\eps^2$ factor.

This packing intuition is also useful for understanding the sample-supported expressibility law. Unsupported expressibility is not an abstract parameter-counting pathology. It means that the ansatz class contains too many mutually distinguishable states for the available copies to resolve. A training algorithm may still choose a state, but the choice cannot be certified uniformly because there are many alternatives that remain statistically plausible.

\section{Adaptive ansatz selection and sample-supported expressibility}\label{sec:adaptive}

In practice, the circuit complexity of the data is not known. Choosing $G$ too small produces a large approximation error; choosing $G$ too large produces a large statistical error and can make the model unsupported by the data. The natural goal is therefore to construct a single learner that adapts to this unknown tradeoff.

\subsection{A hierarchy of circuit classes}

Let
\begin{eqnarray}
1 = G_0 < G_1 < G_2 < \cdots < G_J = G_{\max}
\label{eq:G-hierarchy}
\end{eqnarray}
be a dyadic hierarchy, for example $G_j = 2^j$. The corresponding state classes are nested:
\begin{eqnarray}
\cS_{n,G_0} \subseteq \cS_{n,G_1} \subseteq \cdots \subseteq \cS_{n,G_J}.
\label{eq:S-hierarchy}
\end{eqnarray}
For each $j$, let $\mathcal{H}_j$ be an $\eps_j$-net of $\cS_{n,G_j}$ and let $L_j=\log\abs{\mathcal{H}_j}$ be its description length. By Lemma~\ref{lem:metric-entropy},
\begin{eqnarray}
L_j \le C_{\rm ent}G_j\log{\left(C_{\rm arch}\frac{n G_j}{\eps_j}\right)}.
\label{eq:Lj}
\end{eqnarray}
We also assign a prior weight $\pi_j > 0$ to each model level, with $\sum_j \pi_j \le 1$. A convenient choice is $\pi_j \propto (j+1)^{-2}$, which costs only a $\log j$ penalty in the final bound. These prior weights are not Bayesian assumptions about the data; they are union-bound weights or description-length penalties.

An adaptive learner may now be constructed by taking the union of the finite nets,
\begin{eqnarray}
\mathcal{H}_{\rm all} = \bigcup_{j=0}^J \mathcal{H}_j,
\label{eq:H-all}
\end{eqnarray}
and performing a penalized quantum hypothesis-selection tournament. The candidates from more complex classes are allowed, but they must overcome a larger complexity penalty. This is the quantum analogue of structural risk minimization~\cite{vapnik1998statistical} and minimum description length~\cite{rissanen1978modeling,barron1998minimum}.

A schematic version of the adaptive learner is shown in Algorithm~\ref{alg:adaptive-occam}. The algorithm is not intended as an efficient training routine. Its role is to make the statistical logic explicit. Each circuit class is discretized into a finite net, quantum hypothesis selection is used to compare candidates, and a penalty prevents the learner from choosing a high-complexity candidate merely because the class is large. Appendix~\ref{app:adaptive-proof} provides the detailed oracle-inequality proof behind this construction.
\begin{algorithm}[H]
\caption{Adaptive quantum Occam learner}\label{alg:adaptive-occam}
\begin{algorithmic}[1]
\Require $M$ copies of an unknown state $\hat{\rho}$; dyadic levels $G_j=2^j$; confidence weights $\pi_j$.
\For{$j=0,1,\ldots,J$}
    \State Build a trace-distance net $\mathcal{H}_j$ for $\cS_{n,G_j}$ at radius proportional to $r_j(M,\delta)$.
    \State Assign the level penalty $r_j(M,\delta)$ defined by Eq.~\eqref{eq:rj-implicit}.
\EndFor
\State Let $\mathcal{H}_{\rm all}=\bigcup_{j=0}^{J}\mathcal{H}_j$.
\State Run finite quantum hypothesis selection over $\mathcal{H}_{\rm all}$.
\State \Return a candidate $\hat{\sigma}$ minimizing tournament score plus complexity penalty.
\end{algorithmic}
\end{algorithm}

The penalty can be interpreted as a confidence radius. A large class has a larger net and thus a larger chance of containing a hypothesis that looks artificially good on finite data. The penalty corrects for this multiplicity. In a minimum-description-length interpretation, $L_j=\log\abs{\mathcal{H}_j}$ is the number of bits to specify a hypothesis at level $j$, and the term $\log(1/\pi_j)$ is the number of bits to specify the level itself. The total description length determines the statistical price.

\subsection{Adaptive quantum Occam theorem}

There are several equivalent ways to state the adaptive result. The cleanest form uses an implicit complexity radius. For each $G_j$, define $r_j(M, \delta)$ as the smallest positive number satisfying
\begin{eqnarray}
M r_j^2 \ge C\left[ \log\cN(\cS_{n,G_j},D,r_j/16) + \log({1}/{\pi_j\delta}) \right].
\label{eq:rj-implicit}
\end{eqnarray}
Using Lemma~\ref{lem:metric-entropy}, this radius obeys
\begin{eqnarray}
r_j(M,\delta) = \widetilde{O}\left( \sqrt{\frac{G_j + \log(1/\pi_j\delta)}{M}} \right).
\label{eq:rj-tilde}
\end{eqnarray}

\begin{theorem}[Adaptive quantum Occam theorem]\label{thm:adaptive-occam}
Let $\hat{\rho}$ be an arbitrary $n$-qubit state. There exists an information-theoretic learner, not given the correct circuit complexity, such that with probability at least $1-\delta$ its output $\hat{\sigma}$ satisfies
\begin{eqnarray}
D(\hat{\rho}, \hat{\sigma}) \le C_0 \inf_{0 \le j \le J} \left[ d_{G_j}(\hat{\rho}) + r_j(M, \delta) \right]
\label{eq:adaptive-oracle}
\end{eqnarray}
for a universal constant $C_0>0$. Equivalently,
\begin{eqnarray}
D(\hat{\rho}, \hat{\sigma}) \le C_0 \inf_{G \le G_{\max}} \left[ d_G(\hat{\rho}) + \widetilde{O}\left( \sqrt{\frac{G+\log(1/\delta)}{M}} \right) \right],
\label{eq:adaptive-oracle-tilde}
\end{eqnarray}
where the infimum over $G$ is understood up to the dyadic discretization and logarithmic prior factors.
\end{theorem}

\begin{proof}[Proof sketch]---For each level $j$, take a net $\mathcal{H}_j$ at resolution proportional to $r_j(M, \delta)$ and assign it failure probability $\pi_j \delta$. Applying Lemma~\ref{lem:finite-selection} at level $j$ gives a high-probability agnostic guarantee for that level. A union bound over $j$ makes all level-wise guarantees hold simultaneously with probability at least $1-\delta$. The adaptive learner then selects the candidate minimizing the penalized tournament score, where the penalty is proportional to $r_j(M, \delta)$. On the simultaneous good event, the selected candidate cannot be much worse than the best candidate at any level because any excess empirical fit must be paid for by the corresponding penalty. Taking the infimum over $j$ yields Eq.~\eqref{eq:adaptive-oracle}. Substituting the entropy estimate gives Eq.~\eqref{eq:adaptive-oracle-tilde}. Appendix~\ref{app:adaptive-proof} gives a weighted-union proof that makes the union bound and the role of the prior weights explicit.
\end{proof}

Theorem~\ref{thm:adaptive-occam} is our second main result. It says that a single learner can adapt to the unknown circuit complexity of the data. The learner pays, up to logarithmic factors, the same price it would have paid had the best $G$ been known beforehand. In this sense, $G$ becomes a data-selected resource rather than a fixed promise.

It is helpful to view the theorem as a simultaneous confidence statement over the entire approximation curve. For each $G_j$, the quantity $d_{G_j}(\hat{\rho})$ is not known to the learner. The learner does not estimate this curve explicitly. Instead, the tournament implicitly compares candidates from different levels. The oracle inequality then says that, after the fact, the selected candidate performs almost as well as the best level on the curve. This is the same logic that makes structural risk minimization powerful in classical learning: the procedure does not need to know the correct complexity, because the penalty allows all complexities to be compared on a common statistical scale.

One may also interpret the theorem as a form of quantum model selection under misspecification. If the true state is exactly in a small class, say $\hat{\rho} \in \cS_{n,G_*}$, then $d_{G_*}(\hat{\rho})=0$ and the adaptive learner pays essentially the $G_*$-level penalty. If the true state is not in any small class but is well approximated by a sequence of larger classes, the learner automatically balances approximation and estimation. If the true state is arbitrary and no low-complexity approximation exists, then the infimum moves toward the tomography scale. Thus, the theorem interpolates among realizable low-complexity learning, agnostic compressed learning, and ordinary high-dimensional tomography within a single inequality.

One can regard Eq.~\eqref{eq:adaptive-oracle} as the quantum analogue of a penalized empirical-risk bound. If $\widehat{R}(\hat{\sigma})$ denotes the empirical tournament score of a candidate and $\mathrm{pen}(\hat{\sigma})$ is the radius associated with the smallest level containing it, the learner is conceptually minimizing
\begin{eqnarray}
\widehat{R}(\hat{\sigma}) + \mathrm{pen}(\hat{\sigma}),
\quad
\mathrm{pen}(\hat{\sigma}) \simeq \left[\frac{\Lambda(\hat{\sigma},r)}{M}\right]^{1/2},
\label{eq:penalized-score}
\end{eqnarray}
where
\begin{eqnarray}
\Lambda(\hat{\sigma},r) := \log\cN(\cS_{n,G(\hat{\sigma})},D,r) + \log\frac{1}{\pi_{G(\hat{\sigma})}\delta}.
\label{eq:penalty-lambda}
\end{eqnarray}
The empirical score rewards agreement with the observed copies; the penalty charges for the size of the model class in which the candidate lives. A complex candidate must fit the data sufficiently better than a simple candidate to overcome its larger penalty. This is the formal version of the statement that samples justify expressibility.

The prior weights $\pi_j$ play a modest but conceptually useful role. If all model levels were treated equally without penalty, then a union bound over many possible levels would fail. Choosing $\pi_j \propto (j+1)^{-2}$ assigns a finite total failure probability and adds only a logarithmic correction. In MDL language, this is the cost of specifying the model level before specifying the circuit inside that level. In practical ansatz selection, such a term is analogous to a preference for simpler architectures unless a deeper architecture gives a statistically significant improvement.

Theorem~\ref{thm:adaptive-occam} is robust to the exact discretization of $G$. A dyadic grid loses at most a factor of two in circuit size: if the best gate budget is $G$, then some dyadic level satisfies $G \le G_j \le 2G$. Since all bounds are stated up to universal constants and logarithmic factors, this discretization has no effect on the main scaling. One may replace the dyadic hierarchy by any nested architecture family, such as layers of a hardware-efficient ansatz, provided a metric-entropy bound is available for each level.

Theorem~\ref{thm:adaptive-occam} also gives a clear interpretation of ansatz depth. For a fixed sample budget $M$, the statistically optimal depth is not necessarily the largest implementable depth. Instead, it is the depth minimizing the tradeoff
\begin{eqnarray}
\Phi_M(G; \hat{\rho}) := d_G(\hat{\rho}) + \widetilde{O}\left(\sqrt{\frac{G}{M}}\right).
\label{eq:Phi}
\end{eqnarray}
If $d_G(\hat{\rho})$ decreases rapidly with $G$, a deeper ansatz is justified. If $d_G(\hat{\rho})$ has already saturated, additional gates increase the statistical penalty without improving approximation. This is the quantum Occam principle in operational form.

\subsection{The sample-supported expressibility law}

The adaptive theorem can be turned around. Suppose the target accuracy is $\eps$. For a circuit class of size $G$ to be uniformly learnable at this accuracy, its statistical penalty must be no larger than $\eps$. Ignoring logarithmic factors, this means
\begin{eqnarray}
\sqrt{\frac{G}{M}} \lesssim \eps,
\label{eq:support-condition}
\end{eqnarray}
or equivalently,
\begin{eqnarray}
G \lesssim M\eps^2.
\label{eq:G-supported-basic}
\end{eqnarray}
The lower bound in Corollary~\ref{cor:agnostic-lower} shows that the same relation is necessary, up to logarithmic factors, in the circuit-limited regime. We therefore obtain the third main result.

\begin{corollary}[Sample-supported expressibility law]\label{cor:sample-supported}
At trace-distance accuracy $\eps$ and confidence bounded away from zero, a sample budget $M$ can uniformly support circuit expressibility only up to
\begin{eqnarray}
G_{\rm supported}(M,\eps) = \widetilde{\Theta}\left(M\eps^2\right)
\label{eq:G-supported}
\end{eqnarray}
in the regime $G \ll 2^n$. Including the tomography saturation gives
\begin{eqnarray}
G_{\rm supported}(M,\eps) = \widetilde{\Theta}\left(\min\{2^n,M\eps^2\}\right).
\label{eq:G-supported-saturation}
\end{eqnarray}
Consequently, ansatz classes with $G \gg M\eps^2$ are statistically unsupported for uniform learning at accuracy $\eps$, even when arbitrary collective measurements are allowed.
\end{corollary}

This law is not a statement about physical implementability. A quantum processor may be able to implement a circuit with $G \gg M\eps^2$ gates. Nor is it a statement about whether the target state lies somewhere inside that large ansatz class. Rather, it is a statement about statistical identifiability: with only $M$ copies, the learner cannot uniformly resolve all degrees of freedom made available by such an expressive class.

\begin{table*}[t]
\centering
\setlength{\tabcolsep}{10pt}
\begin{adjustbox}{max width=\textwidth}
\begin{tabular}{lll}
\hline\hline
Regime & Relation & Interpretation \\
\hline
\makecell[l]{Sample-supported} & $G\lesssim M\eps^2$ & \makecell[l]{Uniform trace-distance learning can be certified.} \\
\makecell[l]{Unsupported expressibility} & $G\gg M\eps^2$ & \makecell[l]{Extra ansatz freedom is not resolved by the available copies.} \\
\makecell[l]{Tomography saturation} & $G\gtrsim 2^n$ & \makecell[l]{The full pure-state dimension is the bottleneck.} \\
\hline\hline
\end{tabular}
\end{adjustbox}
\caption{Sample-supported expressibility regimes. The boundary is governed by the statistical radius $\epsilon_{\rm stat}(G,M) \sim \sqrt{G/M}$; hence $G \sim M\eps^2$ at target accuracy $\eps$.}
\label{tab:phase-diagram}
\end{table*}

Table~\ref{tab:phase-diagram} summarizes the law. It should be read as a statistical design rule for circuit-based quantum learning. Given a target accuracy $\eps$ and a copy budget $M$, using an ansatz substantially above $M\eps^2$ gates may still be useful for heuristic optimization, interpolation, or task-specific prediction, but it is not justified by a uniform trace-distance learning guarantee. Conversely, if the data are known to have approximate circuit complexity $C_\eta(\hat{\rho}) \ll M\eps^2$, then the learning problem is statistically feasible at accuracy $\eta+\eps$. Appendix~\ref{app:numerics} visualizes both the approximation-estimation tradeoff and the boundary $G_{\rm supported}=\min\{2^n,M\eps^2\}$.

There is a useful contrapositive. Suppose one insists on using an ansatz of complexity $G$ and wants a global trace-distance guarantee. Then, the required number of copies must scale at least as
\begin{eqnarray}
M \gtrsim \frac{G}{\eps^2}
\label{eq:required-M-for-G}
\end{eqnarray}
up to logarithmic factors, unless $G$ has already saturated the full pure-state manifold. Thus, the same theorem can be used in two directions. Given $M$, it tells us the supported expressibility. Given $G$, it tells us the required sample budget. This dual reading is often the most convenient one for experimental planning.

The law also distinguishes two notions of ``large data.'' In classical machine learning, the number of examples is often compared with the number of parameters. In quantum state learning, the number of copies should be compared with the number of statistically distinguishable quantum hypotheses at the target resolution. A model with many trainable parameters but strong redundancy may have a smaller metric entropy than its raw parameter count. Conversely, a model with few architectural parameters but a large set of discrete choices may have a larger description length than expected. The metric-entropy formulation keeps track of the quantity that matters for learning.

This observation is particularly relevant for hardware-efficient ansatz classes. A hardware-efficient circuit with many layers may have many parameters, but symmetries, parameter tying, limited connectivity, or noise-induced concentration can reduce the effective set of distinguishable output states. Conversely, a family that allows many different qubit-pair layouts may have a larger entropy than a fixed layered architecture with the same number of continuous angles. Thus, in applying the theorem, $G$ should be viewed as a proxy for metric entropy. When a sharper entropy estimate is available, it can be substituted directly into the Occam penalty.

The same logic applies to hybrid classical-quantum models. If a classical preprocessing step selects among many circuit architectures, the description length of that selection contributes to the overall capacity. If a classical neural network controls the parameters of a quantum circuit, then the capacity of the composite model may exceed the raw gate count. The theorem therefore encourages a complete accounting of all choices that enlarge the hypothesis family. What matters is not where the expressibility comes from, but how many distinguishable quantum hypotheses it makes available at the target resolution.

One may also express the law in terms of an effective statistical dimension. Define
\begin{eqnarray}
d_{\rm stat}(G) := \min\{2^n, G\}
\label{eq:dstat}
\end{eqnarray}
up to logarithmic factors. The realizable lower bound says that $M\eps^2 \gtrsim d_{\rm stat}(G)$ is necessary. The upper bound says that it is sufficient up to logs. Thus, $M\eps^2$ plays the role of the number of statistically resolvable circuit degrees of freedom. This is analogous to the way the number of samples in a classical parametric model determines the number of parameters that can be estimated at a target accuracy, but here the parameter count is replaced by the metric entropy of the quantum circuit family.

The saturation at $2^n$ deserves emphasis. If $G$ exceeds the state-synthesis threshold, then increasing $G$ does not create a statistically larger pure-state class, because the class has already become essentially all pure states. The sample complexity is then the tomography complexity. Therefore the complete picture is not that the sample cost grows forever with $G$, but that it grows with $G$ until it reaches the Hilbert-space dimension. In the low-complexity regime relevant for near-term circuit models, the linear-in-$G$ law is the operative one.

\subsection{Overparameterization, underspecification, and quantum data}

Classical machine learning often benefits from overparameterization~\cite{belkin2019reconciling}, and it would be misleading to claim that large quantum ansatz classes are always bad. The message here is more precise. The bounds concern uniform learning of an unknown quantum state in trace distance. They do not rule out favorable distributions, benign inductive biases, task-specific losses, or optimization dynamics that select simple solutions inside a large model. What they do show is that such favorable behavior cannot be guaranteed from expressibility alone.

There is also an underspecification aspect. When $G \gg M\eps^2$, many hypotheses in $\cS_{n,G}$ may be statistically indistinguishable at the available sample size. The training may select one of them, but the choice is not forced by the data. This phenomenon is particularly relevant for quantum generative modeling, where a learned circuit may reproduce certain measurements while remaining far in trace distance from the source. The sample-supported expressibility law identifies the scale at which this ambiguity becomes unavoidable in the worst case.

\section{Implications for quantum machine learning}\label{sec:implication}

The theorems above are information-theoretic, but their interpretation is practical. They convert bounded circuit complexity into a statistical resource and provide a sample-dependent rule for ansatz selection. In this section, we discuss several implications and clarify the relation to prior studies.

\subsection{Quantum generative modeling and variational ansatz design}

Quantum generative models aim to learn a quantum state, a measurement distribution, or a data-generating process using a parameterized quantum circuit~\cite{benedetti2019parameterized,schuld2021machine}. A common design instinct is to make the ansatz expressive enough to cover the target family. Our results imply that this is only half of the design problem. The sample budget supplies an independent upper scale:
\begin{eqnarray}
G_{\rm ansatz} \lesssim M\eps^2
\label{eq:ansatz-rule}
\end{eqnarray}
if the goal is uniform state learning in trace distance at accuracy $\eps$. If an ansatz is much deeper than this scale, then its additional expressibility is not statistically resolved by the data. From the viewpoint of Occam learning, the larger ansatz should be accompanied by either more samples, stronger structural assumptions, or a weaker task-specific loss.

This observation complements the literature on barren plateaus and trainability~\cite{mcclean2018barren,cerezo2021variational}. Barren plateau results ask whether an optimization landscape provides useful gradients. The sample-supported expressibility law asks a different question: even if the optimizer were ideal, are there enough quantum samples to identify the target within the chosen hypothesis class? A model may be trainable but statistically unsupported, or statistically appropriate but hard to optimize. Both constraints matter.

The law can be used as a practical planning rule. Suppose an experiment can prepare $M$ copies of a state and the desired trace-distance accuracy is $\eps$. Then, before choosing a variational architecture, one should compare the intended gate count with $M\eps^2$. If $G$ is below this scale, the ansatz is not ruled out statistically. If $G$ is far above it, then a uniform state-learning guarantee requires either more copies or additional assumptions, such as locality of observables, a known symmetry sector, a restricted distribution over targets, or a task loss weaker than global trace distance. This is not meant to replace empirical validation, but it gives a first-principles scale for whether the expressibility of a proposed ansatz is supported by the data.

\subsection{Quantum data compression}

The approximate circuit complexity $C_\eta(\hat{\rho})$ gives a natural compression interpretation. If $C_\eta(\hat{\rho})=G$, then $\hat{\rho}$ can be represented, up to error $\eta$, by a circuit description of length roughly $G\log G$ bits plus continuous parameter precision. Corollary~\ref{cor:Ceta} says that
\begin{eqnarray}
M = \widetilde{O}\left(\frac{C_\eta(\hat{\rho})}{\eps^2}\right)
\label{eq:compression-law}
\end{eqnarray}
copies suffice to learn such a compressed representation to accuracy $\eta+\eps$. Thus, quantum state learning and quantum data compression are two sides of the same problem. A source is sample-efficiently learnable precisely when it has a short approximate circuit description at the relevant accuracy.

This viewpoint is useful for noisy intermediate-scale devices. Suppose an experimental state is produced by an intended depth-$G$ circuit followed by noise. If the noise is weak in trace distance, then $d_G(\hat{\rho})$ is small and the agnostic theorem gives a meaningful guarantee. If the noise accumulates so that $d_G(\hat{\rho})$ becomes large, then a pure circuit-generated model is no longer appropriate. One should enlarge the hypothesis class to include noisy circuits, mixed states generated by purification, or effective channels. The same Occam logic applies once the metric entropy of the enlarged class is known.

\subsection{Relation to tomography}

The present framework recovers ordinary tomography as a limiting case. When $G$ is large enough to approximate arbitrary $n$-qubit pure states, the effective dimension saturates at $2^n$. The lower bound in Theorem~\ref{thm:realizable-benchmark} then becomes
\begin{eqnarray}
M \gtrsim \frac{2^n}{\eps^2},
\label{eq:tomography-saturation}
\end{eqnarray}
which is the usual pure-state tomography scale. Thus, the circuit-complexity law does not contradict tomography; it interpolates between tomography and low-complexity learning. The relevant statistical dimension is
\begin{eqnarray}
d_{\rm eff}(n,G) \simeq \min\{2^n, G\},
\label{eq:effective-dim}
\end{eqnarray}
up to logarithmic factors.

This interpolation is important for quantum machine learning because it separates two barriers. The Hilbert-space dimension barrier is unavoidable for arbitrary states. The circuit-complexity barrier is the relevant one for compressible states. If a learning problem appears to require exponentially many samples, the reason may be that the data are not known to be compressible, or that the ansatz class is too broad. Conversely, if a state family is generated by circuits of size $G \ll 2^n$, then the tomography barrier is not the correct benchmark~\cite{bang2026learning}.

\subsection{Relation to previous study on low-complexity quantum states}

After deriving the realizable component within the present metric-entropy and packing framework, we note that its scaling is consistent with existing sample-complexity laws for learning states and unitaries of bounded circuit complexity; in particular, Ref.~\cite{zhao2024learning} obtains a linear-in-$G$ dependence, up to logarithmic and accuracy factors, in a related bounded-complexity setting. The present work then goes beyond the realizable comparison by treating circuit complexity as approximate, unknown, and selected by finite samples.

The distinction is threefold. First, our main object is agnostic: the target $\hat{\rho}$ need not lie in $\cS_{n,G}$. The central quantity is the approximation curve $G \mapsto d_G(\hat{\rho})$, or equivalently the approximate circuit complexity $C_\eta(\hat{\rho})$. Second, the circuit complexity is not assumed known. The adaptive theorem gives a model-selection guarantee over a hierarchy of ansatz classes and compares the learner with the best tradeoff over all $G$. Third, the final message is a design law for quantum machine learning: a sample budget $M$ supports only $G \sim M\eps^2$ circuit expressibility at accuracy $\eps$.

This shift matters conceptually. A realizable theorem answers the question, ``How many copies are needed if the target is known to be a $G$-gate state?'' The agnostic-adaptive theory answers a different question, ``Given finite quantum data and no exact model promise, how much circuit complexity can be justified by the samples?'' The latter is closer to the way ansatz classes are used in quantum machine learning.

\subsection{Comparison with the realizable bounded-complexity viewpoint}

\begin{table*}[t]
\centering
\setlength{\tabcolsep}{5pt}
\begin{adjustbox}{max width=\textwidth}
\begin{tabular}{lll}
\hline\hline
Aspect & \makecell[l]{Realizable bounded-complexity learning} & \makecell[l]{Quantum Occam learning in this work} \\
\hline
Target assumption & \makecell[l]{$\hat{\rho} \in \cS_{n,G}$ for a known or promised $G$} & \makecell[l]{Arbitrary $\hat{\rho}$, measured by $d_G(\hat{\rho})$ and $C_\eta(\hat{\rho})$} \\
Main guarantee & \makecell[l]{Learn a $G$-gate state with $\widetilde{\Theta}(G/\eps^2)$ copies} & \makecell[l]{Compete with the best approximation-estimation tradeoff} \\
Role of $G$ & \makecell[l]{Static promise on the target class} & \makecell[l]{Adaptive statistical resource selected by the data} \\
QML interpretation & \makecell[l]{Bounded circuit complexity breaks the tomography barrier} & \makecell[l]{Finite samples determine usable ansatz expressibility} \\
Design message & \makecell[l]{If the target has $G$ gates, $G/\eps^2$ samples are enough} & \makecell[l]{Use ansatz complexity only up to $G \sim M\eps^2$, \\ unless additional structure is available} \\
\hline\hline
\end{tabular}
\end{adjustbox}
\caption{Conceptual difference between the realizable bounded-complexity learning law and the agnostic-adaptive Occam principle developed here. The realizable law is the base theorem; the new QML content comes from treating circuit complexity as an unknown, approximate, and sample-supported quantity.}
\label{tab:comparison}
\end{table*}

The difference between a realizable bounded-complexity theorem and the present Occam-learning viewpoint can be summarized in Table~\ref{tab:comparison}. The linear-in-$G$ sample law is not presented as an isolated new phenomenon. Rather, it is used as the statistical engine for a broader model-selection statement. This distinction is important because many quantum machine-learning problems do not come with a reliable exact circuit-size promise. They come with data, a family of possible ansatz architectures, and uncertainty about how much structure is real rather than noise or finite-sample fluctuation.

\section{Summary and discussion}\label{sec:discussions}

We have developed an information-theoretic Occam theory for circuit-based quantum learning. The starting point was the observation that the expressibility by itself is not a sufficient design principle for quantum machine learning: a circuit family may contain accurate hypotheses, but finite quantum data may not support the statistical resolution needed to identify them. We formalized this tension for the class $\cS_{n,G}$ of $n$-qubit pure states generated by at most $G$ two-qubit gates, with all states and operators written as hatted quantities such as $\hat{\rho}$, $\hat{\sigma}$, and $\hat{U}$. The metric-entropy estimate for this class shows that its logarithmic covering number grows essentially linearly with $G$, up to the architecture-dependent logarithmic factors. Combined with finite quantum hypothesis selection, this yields the realizable benchmark $M_{\rm real}(n,G,\eps,\delta)=\widetilde{\Theta}((G+\log(1/\delta))/\eps^2)$ in the circuit-limited regime $G \ll 2^n$. This benchmark is the statistical backbone of this work.

The main conceptual step was to move beyond exact realizability. For an arbitrary source $\hat{\rho}$, possibly mixed, noisy, or misspecified relative to the chosen pure-circuit model, we introduced the best $G$-gate approximation error $d_G(\hat{\rho})$ and the approximate circuit complexity $C_\eta(\hat{\rho})$. The fixed-complexity agnostic theorem shows that, with $M$ copies, an information-theoretic learner can output a circuit-generated hypothesis whose trace-distance error is bounded by the best approximation error in $\cS_{n,G}$ plus the statistical radius $\widetilde{O}(\sqrt{G/M})$. Thus, the gate budget plays the same structural role as a statistical dimension: increasing $G$ can reduce approximation error, but it also increases the number of distinguishable hypotheses that must be resolved from data.

The adaptive theorem removes the assumption that the correct circuit budget is known in advance. By placing a hierarchy of circuit classes under an Occam penalty, the learner satisfies an oracle inequality that competes with the best point on the full approximation curve $G \mapsto d_G(\hat{\rho})$. In this sense, circuit complexity is not a fixed promise about the target; it is a data-selected resource. The theorem implements a quantum analogue of structural risk minimization and minimum description length: a more expressive ansatz is selected only when its improvement in fit is large enough to pay for its larger metric entropy.

Combining the upper bounds with the matching lower-bound scale gives the sample-supported expressibility law. At trace-distance accuracy $\eps$, a budget of $M$ copies can uniformly support only $G_{\rm supported}=\widetilde{\Theta}(M\eps^2)$ gates in the circuit-limited regime, with saturation at the pure-state tomography scale $2^n$. Equivalently, an ansatz with $G \gg M\eps^2$ may be implementable and may contain the target state, but its full uniform trace-distance expressibility is statistically unsupported by the available copies, even if arbitrary collective measurements are allowed. The law should therefore be read as a sample-dependent design rule: either keep the ansatz within the supported scale, supply additional structure, collect more copies, or relax the learning objective to a weaker task-specific loss.

Our work point to natural extensions. A practically complete theory should combine the Occam capacity term with measurement restrictions, such as, local Pauli measurements~\cite{tony2016pauli}, classical shadows~\cite{huang2020predicting}, or shallow measurement circuits~\cite{bertoni2024shallow,akhtar2023scalable}. It should also quantify the metric entropy of noisy and hardware-constrained ansatz families, where the decoherence, locality, parameter tying, symmetries, or conservation laws may reduce the effective number of distinguishable states. Another important direction is to connect the sample-supported expressibility with trainability: barren plateaus and/or optimization barriers determine whether a statistically appropriate model can be found, while the Occam law determines how much model class can be justified by the available quantum data. These are complementary rather than competing constraints.

The broader implication is that bounded circuit complexity should be viewed as a statistical modeling principle for quantum data. In ordinary tomography, the Hilbert-space dimension sets the relevant scale; in circuit-compressible learning, the metric entropy of the reachable circuit states replaces that dimension until tomography saturation is reached. This changes how one should interpret ansatz expressibility in quantum machine learning. More gates do not automatically mean a better learning model. More gates are useful only when the source has approximation structure that the samples can support. The quantum Occam perspective therefore supplies a bridge between variational-circuit design, quantum generative modeling, and quantum data compression: it tells us not only what a circuit family can express, but what finite quantum data can actually justify learning.

\begin{acknowledgments}
This work was supported by the Ministry of Science, ICT and Future Planning (MSIP) by the National Research Foundation of Korea (RS-2024-00432214, RS-2025-03532992, and RS-2025-18362970) and the Institute of Information and Communications Technology Planning and Evaluation grant funded by the Korean government (RS-2019-II190003, ``Research and Development of Core Technologies for Programming, Running, Implementing and Validating of Fault-Tolerant Quantum Computing System''), the Korean ARPA-H Project through the Korea Health Industry Development Institute (KHIDI), funded by the Ministry of Health \& Welfare, Republic of Korea (RS-2025-25456722), and the Ministry of Trade, Industry and Resources (MOTIR), Korea, under the project ``Industrial Technology Infrastructure Program'' (RS-2024-00466693). We acknowledge the Yonsei University Quantum Computing Project Group for providing support and access to the Quantum System One (Eagle Processor), which is operated at Yonsei University.
\end{acknowledgments}

\appendix

\section{Metric entropy of circuit-generated states}\label{app:metric-entropy}

This appendix gives the detailed covering-number argument used in Lemma~\ref{lem:metric-entropy}. The purpose of the argument is not to optimize constants, but to identify the quantity that controls the statistical capacity of a circuit-generated state class. At resolution $\eta$, a continuous family of circuits becomes a finite family because each gate parameter only needs finite precision. The number of bits needed to specify this finite approximation is the metric entropy.

\subsection{Circuit normal form and gate-count convention}

Fix an architecture and a gate convention. For each layout $\ell$ with at most $G$ two-qubit gate locations, write the corresponding unitary as
\begin{eqnarray}
\hat{U}_{\ell}(\theta) = \hat{W}_{qG}\bigl(\theta_{qG}\bigr) \hat{W}_{qG-1}\bigl(\theta_{qG-1}\bigr) \cdots \hat{W}_1(\theta_1),
\label{eq:app-layout-unitary}
\end{eqnarray}
where $q=O(1)$ allows a constant number of elementary one- and two-qubit factors per counted two-qubit gate. Each elementary factor acts nontrivially on at most two qubits and depends on at most $s=O(1)$ real parameters. The constants $q$ and $s$ depend only on the chosen gate convention.

The statement in Lemma~\ref{lem:metric-entropy} uses the usual bounded-complexity convention in which the single-qubit gates are either counted at constant overhead, absorbed into neighboring two-qubit gate blocks, fixed by the architecture, or drawn from a fixed finite gate set. If an alternative architecture allows an arbitrary independent layer of continuous one-qubit rotations on all $n$ qubits at no cost, then the right-hand side of Eq.~\eqref{eq:metric-entropy-general} should be augmented by an additional term of order $n\log(C/\eta)$. This modification does not change the Occam mechanism; it simply says that every continuous degree of freedom that enlarges the reachable state set must be included in the entropy. Throughout our main text, $G$ denotes the effective circuit description length after this convention has been fixed.

We assume that each elementary gate is parameterized over a compact domain, for example, a product of intervals of length at most a convention-dependent constant. The smoothness on a compact domain implies a uniform Lipschitz constant: for any elementary gate $\hat{W}(\theta)$,
\begin{eqnarray}
\norm{\hat{W}(\theta) - \hat{W}(\theta')}_\infty \le L_0 \norm{\theta - \theta'}_2,
\label{eq:app-elementary-lipschitz}
\end{eqnarray}
where $L_0$ depends only on the elementary gate family.

\subsection{Lipschitz control for a fixed layout}

Let $\hat{U}_{\ell}(\theta)$ and $\hat{U}_{\ell}(\theta')$ be two circuits with the same layout. By taking operator norms and using unitary invariance of the norm, we have
\begin{eqnarray}
\norm{\hat{U}_{\ell}(\theta) - \hat{U}_{\ell}(\theta')}_\infty \le \sum_{a=1}^{qG} \norm{\hat{W}_a(\theta_a) - \hat{W}_a(\theta'_a)}_\infty \le L_0\sum_{a=1}^{qG}\norm{\theta_a-\theta'_a}_2.
\label{eq:app-telescope-bound}
\end{eqnarray}
Thus, the circuit-to-unitary map is Lipschitz with constant $O(G)$ in the product parameter metric as follows directly from the telescoping bound in Eq.~\eqref{eq:app-telescope-bound}.

Now discretize each real parameter with mesh width
\begin{eqnarray}
\Delta = \frac{\eta}{4L_0 qG\sqrt{s}}.
\label{eq:app-mesh-width}
\end{eqnarray}
If every component of every parameter vector is rounded to the nearest grid point, then $\norm{\theta_a - \theta'_a}_2 \le \sqrt{s}\Delta$ for each $a$, and Eq.~\eqref{eq:app-telescope-bound} implies
\begin{eqnarray}
\norm{\hat{U}_{\ell}(\theta) - \hat{U}_{\ell}(\theta')}_\infty \le L_0 qG\sqrt{s}\Delta = \frac{\eta}{4}.
\label{eq:app-unitary-net-radius}
\end{eqnarray}
The particular factor $4$ has no significance; it leaves room for harmless changes in constants when converting from unitaries to states.

For pure states, operator-norm closeness of preparation unitaries implies trace-distance closeness of output states. Let
\begin{eqnarray}
\ket{\psi}=\hat{U}\ket{0^n},
\quad
\ket{\phi}=\hat{V}\ket{0^n}.
\end{eqnarray}
Then,
\begin{eqnarray}
D\bigl(\ketbra{\psi}{\psi},\ketbra{\phi}{\phi}\bigr) = \sqrt{1-|\braket{\psi}{\phi}|^2} \le \norm{\ket{\psi}-e^{i\alpha}\ket{\phi}}_2 \le \norm{\hat{U} - e^{i\alpha}\hat{V}}_\infty
\label{eq:app-vector-to-trace}
\end{eqnarray}
for a suitable phase $e^{i\alpha}$. Dropping the phase can only change constants. Therefore, a unitary $O(\eta)$-net for a fixed layout gives an $O(\eta)$-net for the corresponding set of output pure states. Renormalizing $\eta$ by a universal constant gives a trace-distance $\eta$-net.

\subsection{Counting parameter grid points}

For one elementary parameter vector, compactness of the parameter domain and the mesh width in Eq.~\eqref{eq:app-mesh-width} imply at most $(C/\Delta)^s$ grid points. For all $qG$ elementary factors in a fixed layout, the number of grid points is bounded by
\begin{eqnarray}
\left(\frac{C}{\Delta}\right)^{sqG} \le \left(\frac{C'G}{\eta}\right)^{s'G},
\label{eq:app-fixed-layout-grid}
\end{eqnarray}
where $C'$ and $s'$ absorb $L_0, q, s$ and the constant diameter of the parameter domain. Taking logarithms yields
\begin{eqnarray}
\log{N_{\rm param}(\eta;\ell)} \le C_1G\log\!\left(\frac{C_2G}{\eta}\right).
\label{eq:app-param-entropy}
\end{eqnarray}
This is the continuous-parameter contribution. It is linear in $G$ because there are $O(G)$ parameter blocks, and it contains the factor $\log(G/\eta)$ because each block must be specified to finite precision $O(\eta/G)$.

\subsection{Counting layouts}

The remaining contribution is discrete. In an all-to-all architecture, the two-qubit support of each counted gate can be chosen in at most $\binom{n}{2}\le n^2/2$ ways. Additional discrete choices, such as selecting a gate type from a constant-size library or choosing a location in a fixed finite pattern, contribute only a constant factor per gate. Hence, the number of layouts satisfies
\begin{eqnarray}
N_{\rm layout}(n,G) \le (C_3n^2)^G,
\quad
\log{N_{\rm layout}(n,G)} \le C_4G\log(C_5n).
\label{eq:app-layout-count}
\end{eqnarray}
Multiplying the number of layouts by the number of grid points for each layout and using $0<\eta<1/2$ gives
\begin{eqnarray}
\log{\cN(\cS_{n,G},D,\eta)} &\le& C_1G\log{\left(\frac{C_2G}{\eta}\right)} + C_4G\log{(C_5n)} \nonumber \\
	&\le& C_{\rm ent} G \log{\left(\frac{C_{\rm arch}nG}{\eta}\right)},
\label{eq:app-entropy-full}
\end{eqnarray}
which is Eq.~\eqref{eq:metric-entropy-general}. If the layout is fixed in advance, or if the architecture has only a constant number of possible layouts at each depth, then Eq.~\eqref{eq:app-layout-count} is removed and one obtains the simplified form
\begin{eqnarray}
\log{\cN(\cS_{n,G},D,\eta)} \le C_{\rm ent}G \log\left(\frac{C_{\rm ent}G}{\eta}\right).
\label{eq:app-entropy-fixed-layout}
\end{eqnarray}
This proves Lemma~\ref{lem:metric-entropy}.

The proof also explains the operational meaning of the entropy bound. A circuit with real parameters is not an infinite-capacity statistical object at a fixed accuracy. At trace-distance resolution $\eta$, two parameter choices inside the same mesh cell are statistically interchangeable for the purpose of global state learning, because their output states differ by at most $O(\eta)$. The finite number of cells is precisely the Occam description length used in the sample-complexity bounds.

\section{From finite hypothesis selection to the agnostic Occam bound}\label{app:agnostic-proof}

This appendix expands the proof of Lemma~\ref{lem:finite-selection}, Theorem~\ref{thm:agnostic-occam}, and the lower-bound discussion. The main point is that a finite set of candidate quantum states can be searched with a statistical cost proportional to the logarithm of the number of candidates. After the continuous circuit class is replaced by a trace-distance net, the logarithm of the net size becomes the Occam penalty.

\subsection{Finite agnostic selection}

Let $\mathcal{H}=\{\hat{\sigma}_1,\ldots,\hat{\sigma}_K\}$ be a finite set and let $\hat{\rho}$ be arbitrary. Define
\begin{eqnarray}
a_i := D(\hat{\rho}, \hat{\sigma}_i),
\quad
a_* := \min_{1 \le i \le K} a_i.
\label{eq:app-ai}
\end{eqnarray}
For a pair $(i, j)$, let $\hat{P}_{ij}$ be the Helstrom effect, namely the projector onto the positive part of $\hat{\sigma}_i - \hat{\sigma}_j$~\cite{helstrom1976quantum}; more generally, one may take any effect attaining the trace norm. Then,
\begin{eqnarray}
\Tr{\hat{P}_{ij}(\hat{\sigma}_i-\hat{\sigma}_j)} = D(\hat{\sigma}_i, \hat{\sigma}_j).
\label{eq:app-helstrom-gap}
\end{eqnarray}
For any effect $\hat{E}$, the trace distance dominates the distinguishability induced by $\hat{E}$:
\begin{eqnarray}
\abs{\Tr\hat{E}(\hat{\rho} - \hat{\sigma})} \le D(\hat{\rho}, \hat{\sigma}).
\label{eq:app-effect-contraction}
\end{eqnarray}
Consequently, if candidate $i$ is close to $\hat{\rho}$ and candidate $j$ is much farther away, then the Helstrom comparison between $i$ and $j$ has a detectable bias. Indeed, set
\begin{eqnarray}
p_\rho = \Tr{\hat{P}_{ij}\hat{\rho}},
\quad
p_i = \Tr{\hat{P}_{ij}\hat{\sigma}_i},
\quad
p_j = \Tr{\hat{P}_{ij}\hat{\sigma}_j}.
\end{eqnarray}
Eq.~\eqref{eq:app-effect-contraction} gives $\abs{p_\rho - p_i} \le a_i$. On the other hand, using Eq.~\eqref{eq:app-helstrom-gap} and the triangle inequality,
\begin{eqnarray}
\abs{p_\rho - p_j} &\ge& p_i - p_j - \abs{p_\rho - p_i} \nonumber \\
	&=& D(\hat{\sigma}_i, \hat{\sigma}_j) - \abs{p_\rho - p_i} \nonumber \\
        &\ge& D(\hat{\rho}, \hat{\sigma}_j) - D(\hat{\rho}, \hat{\sigma}_i) - a_i \nonumber \\
        &=& a_j - 2a_i.
\label{eq:app-far-candidate-gap}
\end{eqnarray}
Therefore, if $a_j \ge 3a_i + 4\gamma$, then
\begin{eqnarray}
\abs{p_\rho - p_j} \ge a_i + 4\gamma \ge \abs{p_\rho - p_i} + 4\gamma.
\label{eq:app-comparison-margin}
\end{eqnarray}
An estimate of $p_\rho$ to additive accuracy $\gamma$ is enough to prefer $i$ over $j$ by a robust margin. This is the basic geometric reason why a tournament based on Helstrom comparisons can eliminate candidates that are much farther from $\hat{\rho}$ than the best candidate.

If a single binary comparison is repeated on $m$ independent copies, Hoeffding's inequality gives~\cite{hoeffding1963probability}
\begin{eqnarray}
        \Pr\left[\abs{\widetilde{p}_\rho - p_\rho} > \gamma\right] \le 2\exp(-2m\gamma^2).
\label{eq:app-hoeffding}
\end{eqnarray}
A direct implementation that estimates all $K^2$ pairwise probabilities would be wasteful. The finite-state quantum selection procedures used in Refs.~\cite{haah2016sample,badescu2021quantum} implement the same tournament logic through optimized collective measurements and threshold searches, so that the required number of copies scales as
\begin{eqnarray}
M \ge \frac{C}{\gamma^2} \left( \log K + \log\frac{1}{\delta} \right),
\label{eq:app-finite-selection-samples}
\end{eqnarray}
not as a polynomial in $K$. On the event that all comparisons required by the selection primitive are reliable, every candidate with distance larger than a fixed constant times $a_* + \gamma$ loses to a near-best candidate. The output $\hat{\sigma}$ therefore satisfies
\begin{eqnarray}
D(\hat{\rho}, \hat{\sigma}) \le c a_* + C'\gamma = c\inf_{\hat{\sigma} \in \mathcal{H}} D(\hat{\rho}, \hat{\sigma}) + C'\gamma.
\label{eq:app-finite-selection-bound}
\end{eqnarray}
Renaming $C'\gamma$ as $\eps$ gives Lemma~\ref{lem:finite-selection}. The constants are not important for the main results; only the dependence on $\log K$, $\log(1/\delta)$, and $1/\eps^2$ is used.

It is sometimes useful to state a weighted version. If the candidates carry weights $w_i > 0$ with $\sum_i w_i \le 1$, the same argument with a weighted union bound gives confidence radius
\begin{eqnarray}
r_i(M, \delta) := C\sqrt{\frac{\log(1/w_i) + \log(1/\delta)}{M}}
\label{eq:app-weighted-radius}
\end{eqnarray}
up to universal constants, and an output satisfying
\begin{eqnarray}
D(\hat{\rho}, \hat{\sigma}) \le C_0\inf_i \left[ D(\hat{\rho}, \hat{\sigma}_i) + r_i(M, \delta) \right].
\label{eq:app-weighted-selector}
\end{eqnarray}
This weighted form is the most convenient way to prove the adaptive theorem in Appendix~\ref{app:adaptive-proof}.

\subsection{Reduction from a circuit class to a finite class}

We now pass from a continuous circuit class to a finite net. Fix $G$ and an additive accuracy parameter $\eps$. Let $\mathcal{H}$ be an $(\eps/8)$-net of $\cS_{n,G}$ in trace distance with the net points chosen inside $\cS_{n,G}$. Lemma~\ref{lem:metric-entropy} gives
\begin{eqnarray}
\log\abs{\mathcal{H}} \le C_{\rm ent}G \log{\left(\frac{C_{\rm arch}nG}{\eps}\right)}.
\label{eq:app-net-size-detailed}
\end{eqnarray}
Let $\hat{\sigma}_G^\star \in \cS_{n,G}$ be an almost optimal approximant, so that for arbitrary $\xi>0$,
\begin{eqnarray}
D(\hat{\rho}, \hat{\sigma}_G^\star) \le d_G(\hat{\rho}) + \xi.
\label{eq:app-sigma-star}
\end{eqnarray}
By the net property, there exists $\hat{\eta}_G \in \mathcal{H}$ with
\begin{eqnarray}
D(\hat{\sigma}_G^\star, \hat{\eta}_G) \le \eps/8.
\label{eq:app-net-point}
\end{eqnarray}
The triangle inequality gives
\begin{eqnarray}
D(\hat{\rho}, \hat{\eta}_G) \le D(\hat{\rho}, \hat{\sigma}_G^\star) + D(\hat{\sigma}_G^\star, \hat{\eta}_G) \le d_G(\hat{\rho}) + \xi + \eps/8.
\label{eq:app-best-net-error}
\end{eqnarray}
Thus, the best finite-net error is no larger than the best circuit-class error plus the net radius.

Applying Lemma~\ref{lem:finite-selection} to $\mathcal{H}$ with additive accuracy proportional to $\eps$ gives, with probability at least $1-\delta$,
\begin{eqnarray}
D(\hat{\rho}, \hat{\sigma}_G) &\le& c\inf_{\hat{\eta} \in \mathcal{H}} D(\hat{\rho}, \hat{\eta}) + C'\eps \nonumber \\
	&\le& c d_G(\hat{\rho}) + c\xi + C''\eps.
\label{eq:app-agnostic-bound-detailed}
\end{eqnarray}
Since $\xi>0$ was arbitrary, it can be removed. Absorbing constants into $C_0$ and the additive accuracy gives Eq.~\eqref{eq:agnostic-bound}. The copy condition is exactly Eq.~\eqref{eq:agnostic-M}, because Eq.~\eqref{eq:app-finite-selection-samples} uses $\log{\abs{\mathcal{H}}}$ and Eq.~\eqref{eq:app-net-size-detailed} bounds this logarithm by the circuit entropy.

For the learning-curve form, set $r=r_G(M,\delta)$ to be the smallest positive number satisfying
\begin{eqnarray}
M r^2 \ge C \left[ \log{\cN(\cS_{n,G},D,r/16)} + \log{\frac{1}{\delta}} \right].
\label{eq:app-fixed-radius}
\end{eqnarray}
Repeating the preceding argument with an $r/16$-net gives
\begin{eqnarray}
D(\hat{\rho}, \hat{\sigma}_G) \le C_0\left[ d_G(\hat{\rho}) + r_G(M,\delta) \right].
\label{eq:app-fixed-radius-bound}
\end{eqnarray}
Substituting Lemma~\ref{lem:metric-entropy} into Eq.~\eqref{eq:app-fixed-radius} yields
\begin{eqnarray}
r_G(M, \delta) = \widetilde{O}\left( \sqrt{\frac{G + \log(1/\delta)}{M}} \right),
\label{eq:app-fixed-radius-tilde}
\end{eqnarray}
which is Eq.~\eqref{eq:agnostic-learning-curve}.

\subsection{Realizable lower bound and consistency with bounded-complexity learning}

The lower bound used in Corollary~\ref{cor:agnostic-lower} follows because the agnostic problem contains the realizable one. If $\hat{\rho} \in \cS_{n,G}$, then $d_G(\hat{\rho})=0$. Any agnostic learner that guarantees trace-distance error at most $\eps$ must therefore learn the realizable class $\cS_{n,G}$ to the same accuracy.

The standard packing argument has three steps. First, construct a finite subset $\mathcal{P} \subset \cS_{n,G}$, such that
\begin{eqnarray}
D(\hat{\sigma}, \hat{\sigma}') \ge 2\eps ~~\text{for all distinct}~~ \hat{\sigma}, \hat{\sigma}' \in \mathcal{P},
\quad
\log{\abs{\mathcal{P}}} \ge c\min\{2^n,G\}.
\label{eq:app-packing}
\end{eqnarray}
In the circuit-limited regime $G \ll 2^n$, this can be understood by taking $k$ qubits with $2^k \asymp G$, preparing a large packing of $k$-qubit pure states, and tensoring the remaining $n-k$ qubits with $\ket{0}$. Generic $k$-qubit pure states can be synthesized to fixed accuracy using $O(2^k)$ two-qubit gates, so the packing lies inside a $G$-gate class after adjusting constants. When $G$ exceeds the pure-state dimension scale, the packing size saturates at order $2^n$.

Second, suppose a learner outputs $\hat{\rho}$ with $D(\hat{\rho}, \hat{\sigma}_x) \le \eps/2$ for every packing element $\hat{\sigma}_x$. Because the packing points are separated by $2\eps$, a minimum-distance decoder can identify the index $x$ from $\hat{\rho}$ with constant success probability. Thus, learning the state is at least as hard as transmitting the packing index through the quantum experiment.

Third, Fano's inequality converts successful index recovery into an information requirement. If $X$ is uniform on the packing and the measurement transcript is $Y$, then reliable decoding requires
\begin{eqnarray}
I(X;Y) \ge c'\log{\abs{\mathcal{P}}}.
\label{eq:app-fano-requirement}
\end{eqnarray}
Holevo's bound limits $I(X;Y)$ by the Holevo information of the ensemble of $M$ copies~\cite{holevo1973bounds}. For a local packing at trace-distance scale $\eps$, the Holevo information per copy is $O(\eps^2)$ in the standard construction, so
\begin{eqnarray}
I(X;Y) \le C'M\eps^2.
\label{eq:app-holevo-bound}
\end{eqnarray}
Combining Eqs.~\eqref{eq:app-packing}--\eqref{eq:app-holevo-bound} gives
\begin{eqnarray}
M \ge c''\frac{\min\{2^n,G\}}{\eps^2}.
\label{eq:app-realizable-lower}
\end{eqnarray}
This is the lower-bound scale in Theorem~\ref{thm:realizable-benchmark} and Corollary~\ref{cor:agnostic-lower}.

The same linear-in-$G$ realizable scale was also obtained in a related bounded-complexity learning setting in Ref.~\cite{zhao2024learning}. In the present study, the role of that reference is to confirm that the scale derived above agrees with the existing bounded-complexity learning picture. The logical flow here does not require taking Ref.~\cite{zhao2024learning} as an assumption: the Occam argument first derives the metric-entropy upper bound and the packing lower bound, and the external result is then cited as a consistency check on the realizable benchmark.

\section{Adaptive model selection}\label{app:adaptive-proof}

This appendix gives a more explicit proof of Theorem~\ref{thm:adaptive-occam}. The fixed-$G$ theorem says that one can compete with the best state in a single circuit class. The adaptive theorem tells us that one can compete with the best level in a hierarchy without knowing that level in advance. The only additional ingredient is a weighted union bound, or equivalently, an Occam code length for the model level.

\subsection{Weighted union of finite nets}

For each level $j$, let $\mathcal{H}_j$ be an $r_j/16$-net of $\cS_{n,G_j}$ in trace distance. Let
\begin{eqnarray}
K_j := \abs{\mathcal{H}_j},
\quad
L_j := \log{K_j}.
\end{eqnarray}
Assign level weight $\pi_j>0$ with $\sum_{j=0}^J \pi_j \le 1$. A candidate $\hat{\eta} \in \mathcal{H}_j$ receives the weight
\begin{eqnarray}
w(\hat{\eta}) := \frac{\pi_j}{K_j}.
\label{eq:app-candidate-weight}
\end{eqnarray}
If the same state appears in more than one net, we assign it to the smallest level containing it; this can only reduce the total weight. Then,
\begin{eqnarray}
\sum_{\hat{\eta} \in \mathcal{H}_{\rm all}} w(\hat{\eta}) \le \sum_{j=0}^J\pi_j \le 1.
\label{eq:app-weight-sum}
\end{eqnarray}
The weighted finite-selection bound, Eq.~\eqref{eq:app-weighted-selector}, gives a learner such that, with probability at least $1-\delta$,
\begin{eqnarray}
D(\hat{\rho}, \hat{\sigma}) \le C_0 \inf_{\hat{\eta} \in \mathcal{H}_{\rm all}} \left[ D(\hat{\rho}, \hat{\eta}) + \sqrt{\frac{\log(1/w(\hat{\eta})) + \log(1/\delta)}{M}} \right].
\label{eq:app-weighted-union-bound}
\end{eqnarray}
For $\hat{\eta} \in \mathcal{H}_j$,
\begin{eqnarray}
\log\frac{1}{w(\hat{\eta})} = L_j + \log\frac{1}{\pi_j}.
\label{eq:app-weight-log}
\end{eqnarray}
Therefore, the penalty at level $j$ is governed by $L_j + \log(1/\pi_j) + \log(1/\delta)$.

The implicit definition of $r_j(M, \delta)$ in Eq.~\eqref{eq:rj-implicit} is exactly the fixed-point version of this penalty:
\begin{eqnarray}
M r_j^2 \ge C\left[ \log\cN(\cS_{n,G_j},D,r_j/16) + \log\frac{1}{\pi_j\delta} \right].
\label{eq:app-rj-fixed-point}
\end{eqnarray}
Since $\mathcal{H}_j$ is an $r_j/16$-net, $L_j \le \log\cN(\cS_{n,G_j},D,r_j/16)$, and the square-root penalty in Eq.~\eqref{eq:app-weighted-union-bound} is at most a universal constant times $r_j$.

\subsection{Comparison with the best level}

Fix any level $j$. Let $\hat{\sigma}_j^\star \in \cS_{n,G_j}$ satisfy
\begin{eqnarray}
D(\hat{\rho}, \hat{\sigma}_j^\star) \le d_{G_j}(\hat{\rho}) + \xi
\label{eq:app-adaptive-star}
\end{eqnarray}
for arbitrary $\xi>0$. Because $\mathcal{H}_j$ is an $r_j/16$-net, there exists $\hat{\eta}_j \in \mathcal{H}_j$, such that
\begin{eqnarray}
D(\hat{\sigma}_j^\star, \hat{\eta}_j) \le r_j/16.
\label{eq:app-adaptive-net-point}
\end{eqnarray}
The triangle inequality gives
\begin{eqnarray}
D(\hat{\rho}, \hat{\eta}_j) \le d_{G_j}(\hat{\rho}) + \xi + r_j/16.
\label{eq:app-adaptive-net-error}
\end{eqnarray}
Substituting this particular $\hat{\eta}_j$ into Eq.~\eqref{eq:app-weighted-union-bound} yields
\begin{eqnarray}
D(\hat{\rho}, \hat{\sigma}) \le C_1\left[d_{G_j}(\hat{\rho}) + \xi + r_j(M,\delta) \right].
\label{eq:app-adaptive-level-bound}
\end{eqnarray}
Since this holds for every level $j$ and every $\xi>0$, taking the infimum over $j$ and sending $\xi$ down to $0$ gives
\begin{eqnarray}
D(\hat{\rho}, \hat{\sigma}) \le C_1 \inf_{0 \le j \le J} \left[d_{G_j}(\hat{\rho}) + r_j(M,\delta) \right],
\label{eq:app-adaptive-oracle-detailed}
\end{eqnarray}
which is Eq.~\eqref{eq:adaptive-oracle}.

This proof also clarifies why the adaptive procedure does not overfit merely by taking the union of many nets. A candidate from a large class has a smaller code weight $w(\hat{\eta})$, hence a larger value of $\log(1/w(\hat{\eta}))$. It must therefore improve the fit to $\hat{\rho}$ by more than its additional radius before it can beat a simpler candidate. This is the mathematical content of the Occam penalty.

\subsection{Evaluating the radius}

Using Lemma~\ref{lem:metric-entropy} in Eq.~\eqref{eq:app-rj-fixed-point} gives
\begin{eqnarray}
M r_j^2 \ge C\left[ C_{\rm ent}G_j \log{\left(\frac{16C_{\rm arch}nG_j}{r_j}\right)} + \log{\frac{1}{\pi_j\delta}} \right].
\label{eq:app-rj-entropy}
\end{eqnarray}
Solving this fixed point up to logarithmic factors gives
\begin{eqnarray}
r_j(M,\delta) = \widetilde{O}\left( \sqrt{\frac{G_j + \log(1/\pi_j\delta)}{M}} \right),
\label{eq:app-rj-tilde-detailed}
\end{eqnarray}
which is Eq.~\eqref{eq:rj-tilde}. For the common choice $\pi_j=c/(j+1)^2$, the additional term is
\begin{eqnarray}
\log\frac{1}{\pi_j\delta} = O(\log(j+1) + \log(1/\delta)).
\label{eq:app-prior-cost}
\end{eqnarray}
This is logarithmic in the model level and is negligible compared with the leading $G_j$ term in the regimes where the circuit entropy is large.

Finally, a dyadic hierarchy converts the infimum over $j$ into an infimum over all $G \le G_{\max}$ up to constants. For any integer $G$ with $G_0 \le G \le G_{\max}$, choose $j$ such that
\begin{eqnarray}
G \le G_j \le 2G.
\label{eq:app-dyadic-choice}
\end{eqnarray}
Because the classes are nested, $d_{G_j}(\hat{\rho}) \le d_G(\hat{\rho})$. The radius at level $G_j$ is at most a constant-factor version of the radius at $G$, up to logarithmic factors. Hence Eq.~\eqref{eq:app-adaptive-oracle-detailed} implies
\begin{eqnarray}
D(\hat{\rho}, \hat{\sigma}) \le C_0 \inf_{G \le G_{\max}} \left[ d_G(\hat{\rho}) + \widetilde{O}\left( \sqrt{\frac{G+\log(1/\delta)}{M}} \right) \right],
\label{eq:app-adaptive-continuous-G}
\end{eqnarray}
which is Eq.~\eqref{eq:adaptive-oracle-tilde}.

\subsection{Interpretation of the selected level}

The proof does not require the learner to estimate $d_G(\hat{\rho})$ explicitly. The approximation curve appears only in the oracle comparison. Operationally, the learner compares candidates from all levels using a common penalized scale. If a deeper class reduces the distance to $\hat{\rho}$ by less than its increase in $r_j$, it is not selected by the oracle bound. If it reduces the distance by more than the penalty, the adaptive guarantee allows the learner to move to that deeper class. Thus, the theorem implements the rule
\begin{eqnarray}
\text{\em ``use additional expressibility only when it is supported by samples.''}
\end{eqnarray}
This is the formal adaptive version of the sample-supported expressibility principle used in the main text.

\section{Numerical illustrations}\label{app:numerics}

This appendix explains the purpose, methodology, and interpretation of the numerical illustrations. The plots are not used as assumptions in any theorem. They are deterministic visualizations of the two scales that appear in the theory: the approximation-estimation objective
\begin{eqnarray}
\Phi_M(G;\hat{\rho}) = d_G(\hat{\rho}) + \text{statistical radius} \simeq d_G(\hat{\rho})+\sqrt{G/M},
\label{eq:app-numeric-occam-objective}
\end{eqnarray}
and the sample-supported boundary
\begin{eqnarray}
G_{\rm supported}(M,\eps;n) = \min\{2^n,M\eps^2\}.
\label{eq:app-numeric-supported-boundary}
\end{eqnarray}
All logarithmic factors and architecture-dependent constants are suppressed in the plots.

The first numerical goal is to show why the adaptive oracle inequality has a finite optimal circuit scale. The approximation error $d_G(\hat{\rho})$ is nonincreasing in $G$, while the statistical radius grows like $\sqrt{G/M}$ for fixed $M$. Their sum can therefore have a minimum at an intermediate $G$. This minimum is the statistically justified ansatz size for the given sample budget, not necessarily the largest circuit size that can be implemented on hardware.

The second numerical goal is to visualize the law $G_{\rm supported} \simeq M\eps^2$. This boundary is obtained by requiring the statistical radius to be no larger than the target accuracy:
\begin{eqnarray}
\sqrt{G/M} \lesssim \eps {\quad\Longleftrightarrow\quad} G \lesssim M\eps^2.
\label{eq:app-supported-derivation}
\end{eqnarray}
For very large $G$, the circuit class cannot exceed the effective dimension of the pure-state manifold, so the boundary saturates at $2^n$.

Fig.~\ref{fig:app-tradeoff} uses stylized approximation curves rather than data from a specific quantum device. This is intentional: the purpose is to illustrate how different possible approximation profiles interact with the same Occam penalty. The script chooses a logarithmic grid of gate budgets $G$ and evaluates two monotone decreasing proxy curves,
\begin{eqnarray}
d_G^{\rm exp} &=& d_{\rm floor}^{\rm exp} + A_{\rm exp}\exp[-(G/G_{\rm exp})^{\alpha_{\rm exp}}], \nonumber \\[2pt]
d_G^{\rm poly} &=& d_{\rm floor}^{\rm poly} + A_{\rm poly}\left(1+G/G_{\rm poly}\right)^{-\beta}.
\label{eq:app-poly-proxy}
\end{eqnarray}
The constants are chosen only to make both regimes visible on the same axes. The exponential proxy represents a source for which additional circuit depth quickly captures the relevant structure until a floor is reached. The polynomial proxy represents a source for which approximation improves more gradually with depth.

For each proxy curve, the plotted Occam objective is
\begin{eqnarray}
\Phi_M^{\nu}(G) = d_G^{\nu} + \lambda\sqrt{G/M},
\quad
\nu \in \{\mathrm{exp},\mathrm{poly}\},
\label{eq:app-plotted-objective}
\end{eqnarray}
with $\lambda=1$ in the displayed figure. Changing $\lambda$ would shift the precise minimizer but not the qualitative conclusion. Here, the sample budget is $M=5.0 \times 10^4$, and the vertical dotted line is the supported scale for $\eps=0.20$: hence, 
\begin{eqnarray}
M\eps^2 = 2.0 \times 10^3.
\label{eq:app-Meps2-numeric}
\end{eqnarray}

\begin{figure}[t]
\includegraphics[width=0.65\columnwidth]{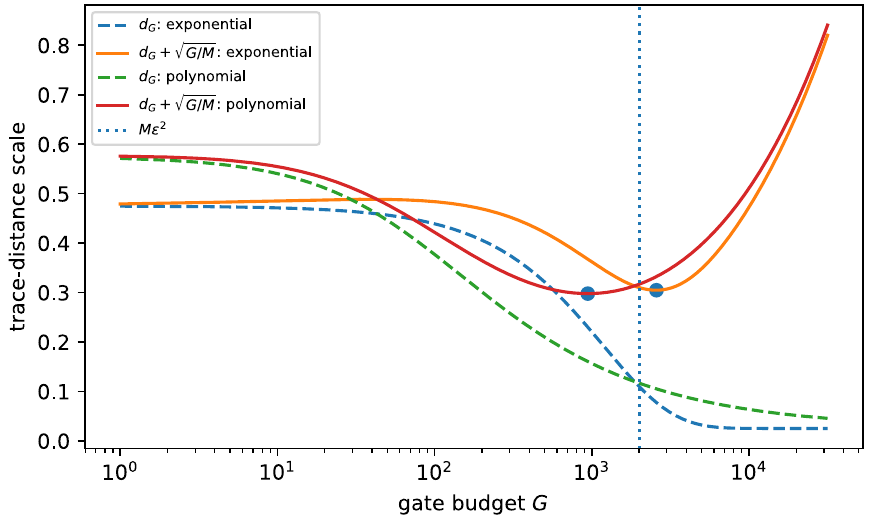}
\caption{Illustrative approximation-estimation tradeoff for $M=5.0\times10^4$. The dashed curves are stylized approximation errors $d_G$, the solid curves add the statistical radius $\sqrt{G/M}$, and the dotted vertical line marks $M\eps^2=2.0\times10^3$ for $\eps=0.20$.}
\label{fig:app-tradeoff}
\end{figure}

The behavior in Fig.~\ref{fig:app-tradeoff} should be read as follows. At small $G$, both solid curves are dominated by approximation error. Increasing $G$ is beneficial because the ansatz is under-expressive. At intermediate $G$, the approximation error and statistical radius are comparable. This is the useful model-selection region: the data can justify increasing the ansatz only if the decrease in $d_G$ is larger than the increase in $\sqrt{G/M}$. At large $G$, the dashed approximation curves continue to decrease or flatten, but the solid Occam objectives rise because the statistical penalty grows. The additional expressibility is then not supported by the available copies for a global trace-distance guarantee.

The exponential and polynomial examples differ in where their minima occur. The exponential proxy drops quickly, so the best objective is obtained soon after the approximation error reaches its floor. The polynomial proxy improves more slowly, so it can justify a larger circuit budget before the statistical penalty dominates. This illustrates the role of the unknown approximation curve $G \mapsto d_G(\hat{\rho})$ in Theorem~\ref{thm:adaptive-occam}. The theorem does not assume either of these shapes; it states that the adaptive learner competes with whichever tradeoff is best for the actual source.

Fig.~\ref{fig:app-boundary} evaluates Eq.~\eqref{eq:app-numeric-supported-boundary} directly. Here, the number of qubits is fixed at
\begin{eqnarray}
n=20, \quad 2^n=1{,}048{,}576,
\label{eq:app-n20}
\end{eqnarray}
and the sample budget $M$ is varied over a logarithmic grid. The target accuracies shown are
\begin{eqnarray}
\eps \in \{0.05,0.10,0.20,0.30\}.
\label{eq:app-epsilon-grid}
\end{eqnarray}
For each pair $(M, \eps)$, the script computes $M\eps^2$ and then takes the minimum with $2^n$.

\begin{figure}[t]
\includegraphics[width=0.65\columnwidth]{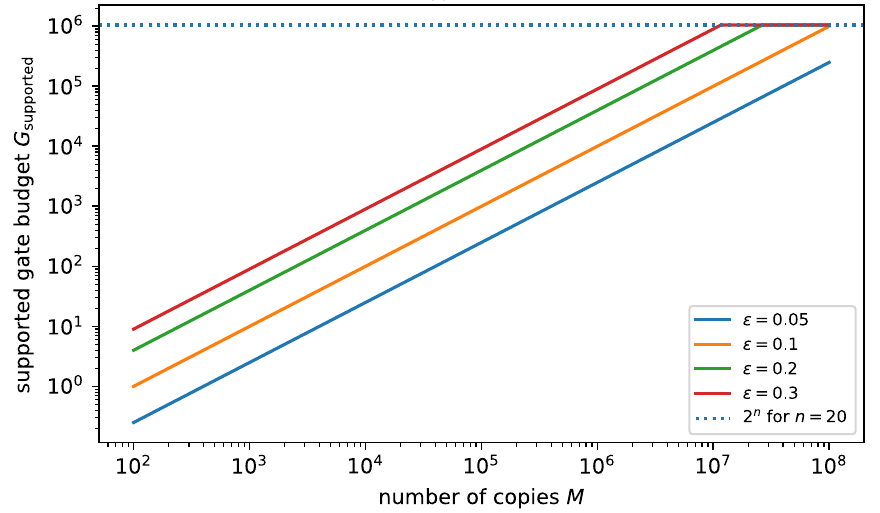}
\caption{Sample-supported expressibility boundary for $n=20$. The curves show $G_{\rm supported}=\min\{2^n,M\eps^2\}$ for several target accuracies, ignoring logarithmic factors. The horizontal dotted line is the tomography saturation $2^{20}$.}
\label{fig:app-boundary}
\end{figure}

The log-log plot has two simple regimes. Below the horizontal plateau, each curve has slope one because
\begin{eqnarray}
\log G_{\rm supported} = \log{M} + 2\log{\eps}.
\label{eq:app-boundary-slope}
\end{eqnarray}
Thus, increasing the number of copies by a factor of ten increases the statistically supported gate budget by a factor of ten at fixed accuracy. Increasing the target accuracy from $\eps=0.10$ to $\eps=0.20$ shifts the curve upward by a factor of four, because the boundary depends on $\eps^2$.

The plateau begins when $M\eps^2$ reaches $2^n$. The saturation sample size is
\begin{eqnarray}
M_{\rm sat}(n,\eps) = \frac{2^n}{\eps^2}.
\label{eq:app-saturation-samples}
\end{eqnarray}
For $n=20$, this gives approximately
\begin{eqnarray}
\begin{array}{c|c}
\eps & M_{\rm sat}(20,\eps) \\
\hline
0.05 & 4.19\times 10^8 \\
0.10 & 1.05\times 10^8 \\
0.20 & 2.62\times 10^7 \\
0.30 & 1.17\times 10^7
\end{array}
\label{eq:app-saturation-table}
\end{eqnarray}
This explains the figure: the larger-accuracy curves reach the $2^{20}$ plateau earlier, while the smaller-accuracy curves continue to grow linearly over most of the displayed range.

Fig.~\ref{fig:app-tradeoff} tells us that increasing expressibility has two effects with opposite signs. It can reduce approximation error, but it also increases the number of statistically distinguishable hypotheses. The lesson of Fig.~\ref{fig:app-boundary} is that the supported gate budget is linear in $M\eps^2$ until the tomography saturation is reached. These conclusions are precisely the qualitative content of Theorem~\ref{thm:adaptive-occam} and Corollary~\ref{cor:sample-supported}~\footnote{Logarithmic factors are deliberately suppressed in both plots, i.e.,  $G \longrightarrow G\log\left(\frac{C_{\rm arch}nG}{\eps}\right)$. It would slightly lower the supported gate budget for fixed $M$ and $\eps$. However, the leading message is unchanged.}.

\bibliographystyle{apsrev4-2}
\bibliography{manuscript_QOL}

\end{document}